\documentclass[final,onefignum,onetabnum]{siamonline190516}

\usepackage{lipsum}
\usepackage{amsfonts}
\usepackage{graphicx}
\usepackage{epstopdf}
\usepackage{algorithmic}
\usepackage{enumitem}
\ifpdf
  \DeclareGraphicsExtensions{.eps,.pdf,.png,.jpg}
\else
  \DeclareGraphicsExtensions{.eps}
\fi

\usepackage{amsmath}
\usepackage{mathrsfs}
\usepackage{amssymb}

\usepackage{tikz}
\usetikzlibrary{arrows.meta}
\usetikzlibrary{trees}
\usetikzlibrary{positioning}

\usepackage{subfig} 

\usepackage{adjustbox}

\newsiamremark{remark}{Remark}
\newsiamremark{hypothesis}{Hypothesis}
\crefname{hypothesis}{Hypothesis}{Hypotheses}
\newsiamthm{claim}{Claim}
\newsiamremark{defi}{Definition}
\newsiamremark{obs}{Observation}
\newsiamthm{coro}{Corollary}
\newsiamthm{prop}{Proposition}
\newsiamthm{axiom}{Assumption}

\newcommand{\defined}{\triangleq}

\newcommand{\X}{\mathcal{X}}
\newcommand{\enu}{\#}

\def\circO{\mathscr{C}}
\definecolor{RED}{RGB}{255,0,0}
\definecolor{csl}{RGB}{0, 0, 0}
\definecolor{green2}{RGB}{0, 200, 45}
\def\bfd{\mathbf{d}}
\def\S{\mathbb{S}}

\newcommand{\csl}[1]{{\color{csl}{#1}}}
\newcommand{\sac}[1]{{\color{csl}{#1}}}
\newcommand{\cgp}[1]{{\color{csl}{#1}}}

\newcommand{\sym}{{\operatorname{Sym}(n)}}
\newcommand{\dihedral}{\operatorname{Dih}}
\newcommand{\diam}{\operatorname{diam}}
\def\id{\mathrm{id}} 
\newcommand{\kendall}{{\tau_K}} 

\newcommand{\cycle}[1]{C_{#1}}

\newcommand{\circular}{{\mathcal{C}_{R}}}
\newcommand{\linear}{{\mathcal{L}_{R}}}
\newcommand{\preC}{{\text{pre-}\circular}}
\newcommand{\preL}{{\text{pre-}\linear}}
\newcommand{\scircular}{{\mathcal{C}_{R}^*}}
\newcommand{\slinear}{{\mathcal{L}_{R}^*}}
\newcommand{\presC}{{\text{pre-}\scircular}}

\newcommand{\rate}[1]{{\mathcal{O}(\log(#1)/#1)}}

\newcommand{\constant}{\ell\delta/(L+\ell)}
\newcommand{\iid}{\stackrel{iid}{\sim}}
\def\circInt{[0, 1)}
\newcommand{\uniform}{\operatorname{Unif}\circInt}

\newcommand{\syminf}{{\operatorname{Sym}(\infty)}}

\newcommand{\x}{x} 
\newcommand{\sample}[1]{{\X_{#1}}}
\newcommand{\sphere}{\mathbb{S}^1}
\newcommand{\sampled}[1]{D_{\X_{#1}}}

\newcommand{\NN}{\operatorname{NN}}
\newcommand{\NNG}{G_{\NN}}
\newcommand{\DFS}{\mbox{DFS}}
\newcommand{\dmin}{\mathbf{d}^{\operatorname{min}}}
\newcommand{\leaves}[1]{{\partial #1}}
\newcommand{\borders}[1]{{\mathcal{B}(#1)}}
\newcommand{\lborders}[1]{{\mathcal{B}^L(#1)}}
\newcommand{\rborders}[1]{{\mathcal{B}^R(#1)}}
\newcommand{\Dmin}{D^{\operatorname{min}}}
\DeclareMathOperator*{\argmin}{\arg\min} 
\DeclareMathOperator*{\argmax}{\arg\max} 
\newcommand{\argdmin}{\mathbf{d}^{\argmin}}
\newcommand{\depth}{ \operatorname{depth}} 

\newcommand{\main}{{\small \texttt{Recursive Seriation}}\,}
\newcommand{\orient}{{\small \texttt{Consecutive Orientation}}\,}
\newcommand{\external}{{\small \texttt{External Orientation}}\,}
\newcommand{\complete}{{\small \texttt{Complete Internal Orientation}}\,}
\newcommand{\completefinal}{{\small \texttt{Final Orientation}}\,}
\newcommand{\partition}{{\small \texttt{Arc Partition}}\,}
\newcommand{\brute}{{\small \texttt{Border Candidates Orientation}}\,}
\newcommand{\depthfirst}{{\small \texttt{Depth-First Search}}\,}

\newcommand{\bordercandidates}{{\small \texttt{Border Candidates}}\,}

\newcommand{\dissimilarity}{{\small \texttt{Tree Dissimilarity}}\,}

\usepackage{enumitem}
\setlist[enumerate]{leftmargin=.5in}
\setlist[itemize]{leftmargin=.5in}

\headers{An optimal algorithm for strict circular seriation}{S. Armstrong, C. Guzmán, C. A. Sing Long}

\title{An optimal algorithm for strict circular seriation\thanks{Submitted to the editors DATE.
\funding{This work was partially supported by INRIA through the INRIA Associate Teams project, CORFO through the Clover 2030 Engineering Strategy - 14ENI-26862, and ANID – Millennium Science Initiative Program – NCN17\_059. C.A.SL. was partially supported by ANID – Millennium Science Initiative Program – NCN17\_129.}}}

\author{Santiago Armstrong\thanks{Institute for Mathematical and Computational Engineering, Pontificia Universidad Cat\'olica de Chile, Santiago, Chile and ANID – Millennium Science Initiative Program – Millennium Nucleus Center for the Discovery of Structures in Complex Data, Santiago, Chile
  (\email{sarmstrong@uc.cl},  \email{casinglo@uc.cl}).}
\and Cristóbal Guzmán\footnotemark[2]~\thanks{Department of Applied Mathematics, University of Twente, The Netherlands (\email{c.guzman@utwente.nl})}
\and Carlos A. Sing Long\footnotemark[2]~\thanks{Institute for Biological and Medical Engineering, Pontificia Universidad Cat\'olica de Chile, Santiago, Chile and ANID – Millennium Science Initiative Program – Millennium Nucleus Center for Cardiovascular Magnetic Resonance, Santiago, Chile.}}

\usepackage{amsopn}

\ifpdf
\hypersetup{
  pdftitle={An optimal algorithm for strict circular seriation},
  pdfauthor={S. Armstrong, C. Guzmán, and C. A. Sing Long}
}
\fi

\begin{document}

\maketitle

\begin{abstract}
  We study the problem of circular seriation, where we are given a matrix of pairwise dissimilarities between $n$ objects, and the goal is to find a {\em circular order} of the objects in a manner that is consistent with their dissimilarity. This problem is a generalization of the classical {\em linear seriation} problem where the goal is to find a {\em linear order}, and for which optimal ${\cal O}(n^2)$ algorithms are known. Our contributions can be summarized as follows. First, we introduce {\em circular Robinson matrices} as the natural class of dissimilarity matrices for the circular seriation problem. Second, for the case of {\em strict circular Robinson dissimilarity matrices} we provide an optimal ${\cal O}(n^2)$ algorithm for the circular seriation problem. Finally, we propose a statistical model to analyze the well-posedness of the circular seriation problem for large $n$. In particular, we establish ${\cal O}(\log(n)/n)$ rates on the distance between 
  any circular ordering found by solving the circular seriation problem to the underlying order of the model, in the Kendall-tau metric.
\end{abstract}



\section{Introduction}

The seriation problem seeks to order a sequence of $n$ objects from pairwise dissimilarity information. The goal is for the objects to be linearly ordered according to their dissimilarity \cite{liiv2010seriation,recanati2017spectral,recanati2018reconstructing}.
Seriation has found applications in several areas such as archaeology \cite{robinson1951method}, sociology and psychology \cite{liiv2010seriation}, and gene sequencing and bioinformatics \cite{recanati2017spectral, mirkin1984graphs}. However, in many applications the objects may be arranged along a closed continuum, resulting in a circular order instead. For instance, in \textit{de novo} genome assembly of bacterial plasmids, the goal is to reorder DNA fragments sampled from a circular genome \cite{recanati2017spectral,liao2019completing}. In some problems in planar tomography, an object's density is to be reconstructed from projections taken at unknown angles between $0$ and $2\pi$. Reordering the projections according to their angle enables the reconstruction of the density \cite{coifman2008graph}. In this case, the matrix representation of the pairwise dissimilarities is symmetric, with entries that increase monotonically  starting from the diagonal along each row  until they reach a maximum and then decrease monotonically, when the columns are wrapped around (see \Cref{fig:ejemprob}). Matrices of this form are called circular Robinson \cite{evangelopoulos2020circular, hubert1998graph} in contrast to linear Robinson dissimilarities, where the entries are monotone non-decreasing along rows and columns when moving toward the diagonal \cite{laurent2017similarity}.

\subsection{Our contributions}

In this work, we address the problem of circular seriation. \csl{In particular, we address algorithmic questions and study its well-posedness for large \(n\)}. Some of our results also apply to the linear case. Our first contribution is to provide a tractable and natural definition of circular Robinson matrices by leveraging unimodality (cf. \cref{prop:rob_is_unimod}). Various definitions of circular ordering have been proposed in the literature (see \Cref{sec:related_work} below), but we believe this one captures intuitively the behavior of circular data.

Our second contribution is to provide the first optimal algorithm, \cgp{i.e.,} with ${\cal O}(n^2)$ time and space complexity, for the seriation problem \cgp{on} strict Robinson dissimilarity matrices. Our algorithm is based on known techniques and data structures used in combinatorial seriation, but by virtue of the strict Robinson property our algorithm is substantially simpler. At a high level, the algorithm follows a divide-and-conquer approach, where we recursively detect nearest neighbors between chains of consecutive elements, and then resolve the orientations of such chains by comparing elements from their borders.

Our third contribution is a statistical model for 
the large \(n\) regime. In this model, points are sampled from a closed curve, which \sac{without loss of generality} we assume is the unit circle, with a continuous and strict circular Robinson dissimilarity. Our main result here is \cgp{an} $\rate{n}$ bound on the expected Kendall-tau distance of any strict circular Robinson ordering of the data. This result is based on an observation 
that in the continuous model, there is essentially\footnote{In an infinite set, permutations can be identified with bijections. However, given that for any finite sample we would only observe permutations of finitely many elements, we can substantially reduce the number of relevant permutations for this question. See \Cref{sec:generative} for further details.} a unique ordering which makes the dissimilarity continuous and strictly circular Robinson. This analysis bridges the gap between solutions to the seriation problem, and their accuracy when data is naturally embedded in a continuous circular-like structure.

\subsection{Related Work} \label{sec:related_work}

Linear seriation is a classical problem in unsupervised learning and exploratory data analysis. As such, it has been thoroughly studied, optimal algorithms for combinatorial seriation are known, as well as spectral methods. In contrast, circular seriation is substantially less understood. Next we summarize some results from the literature.

\paragraph{Linear Seriation}

The first polynomial time algorithm for retrieving a linear order from permuted linear Robinson matrices was due to Mirkin and Rodin \cite{mirkin1984graphs}. It is based on the connection between linear Robinson matrices and interval hypergraphs. It uses an algorithm introduced in \cite{fulkerson1965incidence} as a core subroutine, with an overall running time of $\mathcal{O}(n^4)$. Chepoi and Fichet \cite{chepoi1997recognition} later introduced a simpler algorithm using a divide-and-conquer strategy. By recursively performing a partition refinement the algorithm computes an ordering in $\mathcal{O}(n^3)$ operations and $\mathcal{O}(n^2)$ space. Using similar techniques, Seston \cite{seston2008dissimilarites} improved the complexity to $\mathcal{O}(n^2\log(n))$. Atkins \cite{atkins1998spectral} presented an entirely different strategy based on Laplacian eigenmaps (see \cite{belkin2003laplacian}) with running time of ${\cal
O}(n(T(n)+n \log n)),$ where $T(n)$ is the complexity of (approximately) computing the leading eigenvector of a $n \times n$ symmetric matrix. Prea and Fortin in \cite{prea2014optimal} presented an optimal $\mathcal{O}(n^2)$ algorithm, using an algorithm from \cite{booth1976testing} to first compute a $PQ$-tree which is then updated by the algorithm. For the sparse case, Laurent and Seminarotti \cite{laurent2017similarity} present the Similarity-First Search algorithm \cgp{with}  $\mathcal{O}\left(n^{2}+n m \log n\right)$ operations, where $m$ is the number of nonzero entries of the dissimilarity matrix. 

A natural question is how to perform seriation under noisy measurements of a dissimilarity. Here, it is known that projecting a dissimilarity on the class of Robinsonian dissimilarities (in $\ell_{\infty}$-norm) is an NP-hard problem \cite{chepoi2009seriation}, and constant factor approximation algorithms exist \cite{chepoi2011seriation}.

\paragraph{Circular Seriation}
In contrast to the linear case, where there is a common consensus for the definition of linear Robinson dissimilarities, in the circular case many definitions have been proposed that, in spite that they follow the same intuition, have mathematical formulations that are not equivalent. The first generalization of Robinson dissimilarities to the circular case was introduced in \cite{hubert1998graph}. On top of being quite involved, this definition allows bimodality within each row (modulo $n$), which is incompatible with a circle embedding. The approach proposed for circular seriation is an instance of the quadratic assignment problem, which is NP-hard. 
A recent work following a similar line is \cite{evangelopoulos2020circular}. The authors propose an optimization framework where they employ a spherical embedding together with a spectral method for circular ordering in order to recover circular arrangements of the embedded objects. This heuristic has no theoretical guarantees. A different approach in \cite{coifman2008graph} aims to generalize Atkins' spectral approach by considering two eigenvectors. This methodology has asymptotic guarantees due to the connection between the Laplacian operator and the continuous Laplace-Beltrami operator over a manifold. Using the \sac{same idea, in \cite{recanati2018reconstructing}} theoretical guarantees for a spectral method are introduced for the particular case in which the circular Robinson matrix is circulant, which is an idealized setting. In the same work, numerical experiments are presented to illustrate how the spectral method gains robustness by leveraging higher (\(>2\)) Laplacian eigenvectors. In \cite{brucker2008hypercycles} dissimilarities whose ball, 2-ball and cluster hypergraph correspond to an arc hypergraphs are studied. Such dissimilarities can be considered as generalizations of Robinson dissimilarities to the circular case. We build upon this work by considering dissimilarities whose ball hypergraph corresponds to arcs and connect it to other definitions by showing that this definition is equivalent to requiring that the map $j\mapsto D(i,j+i\bmod n)$ is unimodal. Brucker and Osswald in \cite{brucker2008hypercycles} mainly focus in what they call circular dissimilarities which are a particular case of the previous definition.

\subsection{Outline}

The paper is organized as follows. \Cref{sec:prelim} introduces the notation and preliminaries. In \Cref{sec:seriation_dissim} we formally introduce the seriation problem and the crucial concept of Robinson dissimilarities and matrices. In \Cref{sec:consecutiveOnes} we present some classical results on the {\em consecutive ones problem} and its connection to seriation, including the {\em $PQ$-tree data structure}, which is critical for our optimal algorithm. In \Cref{sec:algorithm} we present our optimal algorithm for strict circular seriation. Finally, in \Cref{sec:generative} we provide the generative model of sampling from a {\em continuous strictly Robinson curve}. 

\section{Preliminaries}
\label{sec:prelim}

\newcommand{\gen}[1]{\langle #1\rangle}

Throughout this work, arrays are indexed starting from $0$ and are real unless it is explicitly stated otherwise. We let $[n]\defined \{0,1,\dots, n-1\}$ and denote as $\sym$ the group of permutations of $[n]$. A permutation is represented either by a vector $\pi$ with entries in $[n]$ or by an \(n\times n\) orthogonal and \cgp{$\{0,1\}$-}matrix $\Pi$. We denote as \(\pi_r\) the permutation that {\em reverses} the elements of \([n]\), i.e., \(\pi_r(i) = n-1-i\), and \(\pi_s\) the {\em cyclic (right) shift} on \([n]\), i.e., \(\pi_s(i) = i + 1\bmod n\). We consider the {\em action by conjugation} of \(\sym\) over the set of \(n\times n\) matrices, which is defined by \((\Pi, A)\mapsto \Pi A\Pi^T\). If \(S\subset \sym\) we denote \(\gen{S}\) the subgroup generated by the elements of \(S\). Finally, we denote the dihedral group of $2n$ different symmetries of a regular polygon with $n$ sides as $\dihedral_n$.

For a countable set \(\X\) and an enumeration \(x: i \to x(i)\) we write \(x_i\) to denote \(x(i)\) and let \(\enu x\) be the integer such that \(x(\enu x) = x\). In this work, we consider finite sets of cardinality \(n\). \sac{An enumeration becomes a bijection \([ n] \mapsto \X\) with inverse \(\enu:\X\to [n]\).}

The notion of \cgp{an} {\em ordered set} will play a crucial role. A {\em linear order} on \(\X\) is a relation \(\leq\) on \(\X^2\) that is reflexive, antisymmetric, transitive and total. The pair \((\X,\leq)\) is a {\em linearly ordered set}. We say \(x_0,\ldots, x_{N-1}\) are {\em linearly ordered} if \(x_i\leq x_{i+1}\) for \(i\in[N]\). A {\em cyclic order} on \(\X\) is a relation \(\circO\) on \(\X^3\) that is cyclic, antisymmetric, transitive and total. The pair \((\X,\circO)\) is a {\em cyclically ordered set}. A cyclic order induces a linear order on \(\X\). For \(x_0\in \X\) we define the linear order \(\leq_{\circO,x_0}\) as \(x \leq_{\circO,x_0} y\) if and only if \((x_0, x, y) \in \circO\). Finally, we say \(x_0,\ldots, x_{N-1} \in \X\) are {\em cyclically ordered} if \(x_i \leq_{\circO,x_0} x_{i+1}\) for \(i\in [N]\). See \cite{novak1982cyclically} for more details.

\section{The seriation problem and Robinson dissimilarities} \label{sec:seriation_dissim}

We introduce the seriation problem. 
Given a set $\mathcal{M}$ of $n\times n$ real matrices, let the {\em pre-$\mathcal{M}$ class} be the orbit of \(\mathcal{M}\) under the action of \(\sym\) by conjugation. The {\em abstract seriation problem} can be stated as \cite{recanati2018reconstructing}

\vspace{4pt}
\begin{center}
    {\em Given $A$ in pre-$\mathcal{M}$ find $\Pi$ in $\sym$ such that $\Pi A\Pi^T$ is in $\mathcal{M}$.}
\end{center}
\vspace{4pt}
The seriation problem is determined by the class $\mathcal{M}$. A {\em solution} to the seriation problem for \(A\) is any permutation \(\Pi\) satisfying the above. \cgp{We denote the set of all solutions by} \(\mathcal{S}_{\mathcal{M}}(A)\).



We study two questions about this problem: \csl{{\em for which class \(\mathcal{M}\) can we ensure a solution exists?}} and, given this class, {\em is there an efficient algorithm to solve the seriation problem for any \(A\)?} \csl{In this work we focus on the case where the class \(\mathcal{M}\) is induced by a dissimilarity on a finite set \(\X\). Our goal is to provide an answer when this dissimilarity may induce a linear or cyclic order on this set.} For this reason, we explicitly distinguish between the {\em linear seriation problem} and the {\em circular seriation problem}; the {\em seriation problem} refers to either of them.


To answer the first question, in \Cref{sec:seriation_dissim:rob_dissim} we characterize dissimilarities that admit such linear or cyclic orders, and in \Cref{sec:seriation_dissim:rob_mat} we discuss how these induce a suitable class of matrices for the seriation problem. We defer the answer to the second question to \Cref{sec:consecutiveOnes} and~\Cref{sec:algorithm}.



\subsection{Robinson dissimilarities}
\label{sec:seriation_dissim:rob_dissim}

\def\dfun{\mathbf{d}}
\def\R{\mathbb{R}}

\cgp{Most of the work on the seriation problem has focused on matrix dissimilarities. In this work, we primarily focus on an equivalent formulation based on abstract pairwise dissimilarities. This allows for a more transparent presentation of algorithms, as well as a more natural extension to the case of infinite continuous sets. Later, we provide a formal connection between the two concepts.}

A {\em dissimilarity} or {\em premetric} \(\dfun:\X^2\to \R\) on \(\X\) is a non-negative and symmetric function 
that is identically zero on the diagonal. {\em Robinson dissimilarities} are 
dissimilarities to which we can associate a linear or cyclic order on \(\X\).

\subsubsection{Linear Robinson dissimilarities}
\label{sec:seriation_dissim:rob_dissim:linear}

\def\HGr{\mathcal{H}}
\def\leqCxo{\leq_{\circO,x_0}}
\def\leCxo{<_{\circO,x_0}}

Linear Robinson dissimilarities admit a family of linear orders on \(\X\).

\begin{defi}[The linear Robinson property] \label{defi:RobDiss}
    A dissimilarity \(\dfun\) on \(\X\) is {\em linear Robinson} if there exists a linear order \(\leq_{\dfun}\) on \(\X\) such that
    \begin{equation}
    \label{eq:abslinrob}
        \forall\,\, \mbox{linearly ordered \(x,y,z\in \X\)}:\,\,
        \dfun(x, z) \geq \max\{\dfun(y, x), \dfun(y, z)\}.
    \end{equation}
    It is {\em strictly linear Robinson} if all the inequalities are strict. We say \(\leq_{\dfun}\) is {\em consistent} with \(\dfun\) and that \(\dfun\) is linear Robinson with respect to \(\leq_{\dfun}\).
\end{defi}




Linear Robinson dissimilarities preserve the {\em intervals} defined by {\em any} consistent order~\cite{mirkin1984graphs}. From \Cref{defi:RobDiss} it follows that for any \(r >0\) and \(x\in\X\) the {\em (closed) balls} $B_r^{\dfun}(x) \defined \{y\in \X:\,\, \dfun(x,y) \leq r\}$ are intervals in \((\X,\leq_{\dfun})\). In fact, this property uniquely characterizes linear Robinson dissimilarities. To prove this converse, the appropriate structure to analyze is the {\em hypergraph} \(\HGr_{\dfun}\) with vertex set \(\X\) and hyperedge set \(\mathbf{B}_{\dfun}\defined \{B_r^{\dfun}(x):\,\, x\in \X,\, r > 0\}\). This hypergraph is called an {\em interval hypergraph} if every hyperedge is an interval~\cite{kobler2017circular}.

\begin{prop}[\cite{mirkin1984graphs}]
\label{prop:lin}
Let \(\dfun\) be a dissimilarity on \(\X\). The following are equivalent:
\begin{enumerate}
    \item $\dfun$ is linear Robinson.
    \item The hypergraph \(\HGr_{\dfun}\) is an interval hypergraph.
\end{enumerate}
\end{prop}




\cgp{It is easy to see that orderings consistent with a given linear Robinson dissimilarity are never unique.} 
This follows from the {\em natural symmetries} of Robinson dissimilarities. Let \(\leq_{\dfun}\) be consistent with respect to \(\dfun\). Its {\em reversal} \(\leq'_{\dfun}\) is the linear order defined by \(x\leq_{\dfun}' y\) if and only if \(y\leq_{\dfun} x\). It is clear that \(\dfun\) is linear Robinson with respect to \(\leq_{\dfun}\) if and only if it is so with respect to \(\leq'_{\dfun}\).  
\cgp{Hence, the reversal of an ordering consistent with a dissimilarity is always consistent with the dissimilarity. Furthermore, there could be other consistent orderings in the case the dissimilarity is not strict, as we will see in \Cref{sec:consecutiveOnes}.}



\subsubsection{Circular Robinson dissimilarities}
\label{sec:seriation_dissim:rob_dissim:circ}

Circular Robinson dissimilarities a\-rise naturally when we allow for cyclic orders.

\begin{defi}[The circular Robinson property]
A dissimilarity $\dfun$ on \(\X\) is {\em circular Robinson} if there exists a cyclic order \(\circO_{\dfun}\) such that
\sac{
\[
    \forall\,\, \mbox{cyclically ordered \(w,x,y,z\in \X\)}:\quad  \dfun(y,w)\geq \min\{\dfun(y,x), \dfun(y,z)\}.
\]}
We say it is {\em strict circular Robinson} if the inequality is strict. We say \(\circO_{\dfun}\) is {\em consistent} with \(\dfun\) and that \(\dfun\) is linear Robinson with respect to \(\circO_{\dfun}\).
\end{defi}

Circular Robinson dissimilarities preserve the {\em arcs} of {\em any} compatible order, i.e., sets of the form \(\{x\in \X:\,\, (m, x, M) \in \circO_{\dfun}\}\) for \(m, M\in \X\) called the {\em borders} of the arc. Arcs are the natural analogues of intervals for a cyclic order. Consequently, we say \(\HGr_{\dfun}\) is an {\em arc hypergraph} if all its hyperedges are arcs. 
The analog of \cref{prop:lin} for a cyclic order is the following. 

\begin{prop}\label{prop:arc_is_ball}
(\cite[Proposition 5]{brucker2008hypercycles})
Let \(\dfun\) be a dissimilarity. 
The following are equivalent:
\begin{enumerate}
    \item \(\dfun\) is circular Robinson.
    \item The hypergraph \(\HGr_{\dfun}\) is an arc hypergraph.
\end{enumerate}
\end{prop}


Similarly to the linear case, consistent orderings in the circular case are never unique. 
Let $\circO_{\dfun}$ be consistent with respect to \(\dfun\). In this case, its {\em reversal} \(\circO_{\dfun}'\) is the cyclic order such that \((x,y,z) \in \circO_{\dfun}'\) if and only if \((z,y,x)\in \circO_{\dfun}\). By definition, \(\dfun\) is circular Robinson with respect to \(\circO_{\dfun}\) if and only if it is so with respect to \(\circO_{\dfun}'\).





\subsection{Robinson matrices}
\label{sec:seriation_dissim:rob_mat}

\begin{figure}
    \centering
    \includegraphics[width=4cm]{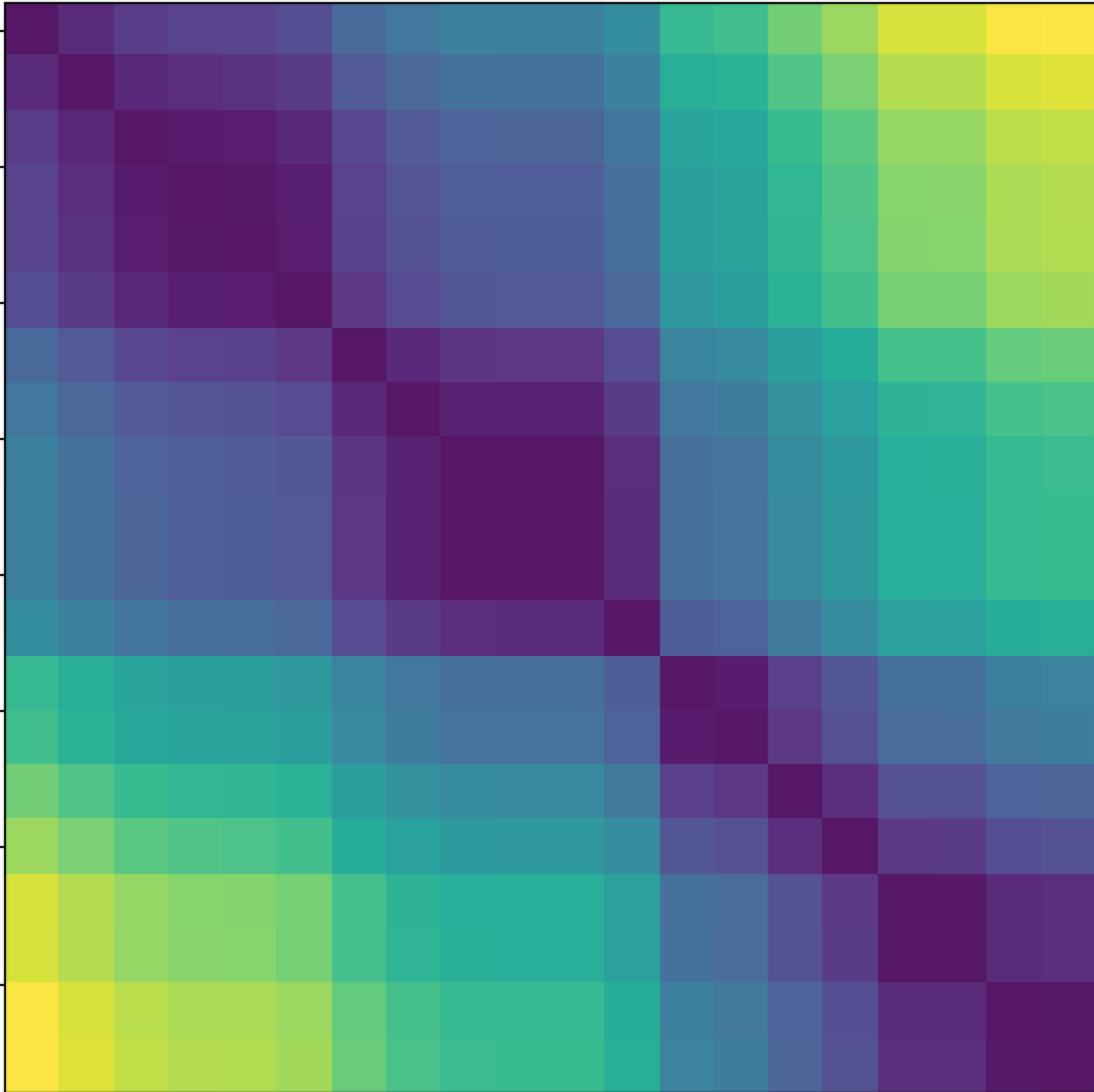}
    \qquad
    \includegraphics[width=4cm]{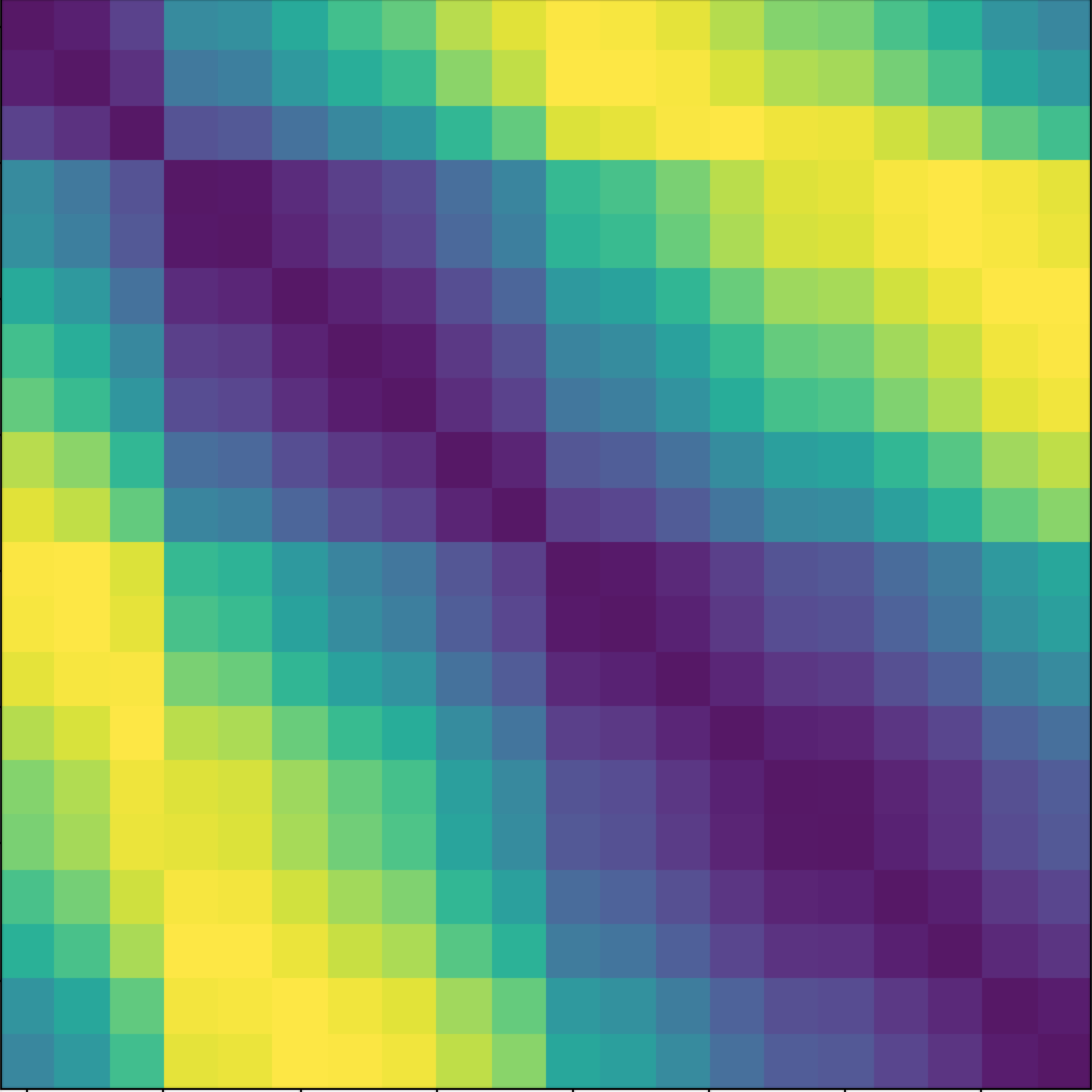}
    \caption{Example of a linear Robinson dissimilarity matrix (in the left) and a circular Robinson dissimilarity matrix (in the right)}\label{fig:ejemprob}
\end{figure}

Let \(\dfun\) be a dissimilarity on \(\X\). To any enumeration \(x:[n] \to \X\) we can associate the \(n\times n\) {\em dissimilarity matrix} \(D\) with entries \(D(i, j) := \dfun(x_i, x_j)\). It is always non-negative, symmetric, and with zero-diagonal. However, some enumerations will endow \(D\) with additional properties. This leads us to \cgp{the concept of} {\em Robinson matrices}.

\subsubsection{Linear Robinson matrices}

If \(\dfun\) is consistent with respect \cgp{to} \(\leq_{\dfun}\) there exists an enumeration of \(\X\) such that for \(i,j\in [n]\) we have \(i\leq j\) if and only if \(x_i \leq_{\dfun} x_j\). In this case, it follows that \(D\) induces a linear Robinson dissimilarity on \([n]\). 

\begin{defi}[Linear Robinson matrix]
\label{RM}
    A dissimilarity 
    matrix \(D\) is {\em linear Robinson} if
    \begin{equation}
    \label{eq:linrob}
        \forall\,\, \mbox{linearly ordered \(i,j,k\in [n]\)}:\,\,
        D(i,k) \geq \max\{D(j,i), D(j,k)\}.
    \end{equation}
    It is {\em strictly linear Robinson} if all the inequalities are strict.
\end{defi}

This implies \(D\) is consistent with the standard order on \([n]\) and \Cref{prop:lin} holds for \(D\) when \(\X = [n]\). 

When the dimension is understood from context, the set of linear and strictly linear Robinson dissimilarity matrices will be denoted $\linear$ and $\slinear$ respectively. Considering each one of these sets as \(\mathcal{M}\) leads to {\em linear seriation} and {\em strict linear seriation}, respectively. 

Note linear Robinson matrices inherit the symmetries from the dissimilarity. In fact, it can be verified that \(D\) is linear Robinson if and only if \(\Pi_r D\Pi_r^T\) is linear Robinson. \cgp{Since} $\dihedral_1\cong \gen{\pi_r}$ \cgp{then} linear Robinson matrices are invariant under the action of \(\dihedral_1\) by conjugation.

\subsubsection{Circular Robinson matrices}

For the circular case, we endow the set \([n]\) with the standard cyclic order \(\circO_n\)
\begin{equation}
\label{eq:circOn}
    (i,j,k)\in \circO_n \quad\Longleftrightarrow\quad (i<j<k) \,\,\vee\,\, (j<k<i) \,\,\vee\,\, (k<i<j).
\end{equation}
We still denote the standard linear order in \([n]\) as \(\leq\).

If \(\dfun\) is consistent with respect \cgp{to} \(\circO_{\dfun}\) there exists an enumeration of \(\X\) such that \((i,j,k)\in \circO_{n}\) if and only if \((x_i, x_j, x_k) \in \circO_{\dfun}\). Similarly to the linear case, this implies \(D\) induces a circular Robinson dissimilarity on \([n]\).

\begin{defi}[Circular Robinson matrix]\label{CRM}
A dissimilarity matrix $D$ is {\em circular Robinson} if
\sac{
\[
    \forall\,\, \mbox{cyclically ordered \(i,j,k,\ell\in [n]\)}:\quad  D(i,k)\geq \min\{D(j,k), D(k,\ell)\}.
\]}
\end{defi}

Therefore, \(D\) is consistent 
w.r.t.~\(\circO_n\) and \Cref{prop:arc_is_ball} holds for \(D\) when \(\X = [n]\). When the dimension is understood from context, the set of circular and strictly circular Robinson matrices \cgp{is denoted by}  $\circular$ and $\scircular$, respectively. Considering each one of these sets leads to the {\em circular seriation problem} and the {\em strict circular seriation} respectively. Comparing \cref{RM} and \cref{CRM}, it is apparent that every linear Robinson matrix is also circular Robinson. In this sense, the notion of circular Robinson extends that of linear Robinson. The difference between linear and circular Robinson matrices is illustrated in~\cref{fig:ejemprob}.

We provide an alternative definition for circular Robinson matrices that will be useful in what follows. Let \(f:[n]\to \R\). A {\em mode} is any \(m\in [n]\) such that
\[
    \forall\, i,j\in [n]:\,\,\cgp{\big(}\mbox{\(i \leq j \leq m\) or \(m \leq j \leq i\)}\cgp{\big)}\,\,\Rightarrow\,\, f_i\leq f_j \leq f_m.
\]
We say \(f\) is {\em unimodal} if it has a mode. We say \(f\) is {\em strictly unimodal} if it has at most two distinct, consecutive modes \(m_1 \leq m_2\) with \(f_{m_1} = f_{m_2}\) and
\[
    \forall\, i,j\in [n]:\,\, \cgp{\big(}\mbox{\(i < j < m_1\,\,\Rightarrow\,\, f_i< f_j< f_{m_1}\) \cgp{\big)} and \cgp{\big(}\(i > j > m_2\,\,\Rightarrow f_i < f_j < f_{m_2}\cgp{\big)}\).}
\]
\csl{Notice that when \(m_1 < m_2\) there are two modes, whereas when \(m_1 = m_2\) there is only one mode. Our definition above allows us to treat both cases simultaneously.}
From the definition it is clear that every subsequence of a strictly unimodal sequence is also strictly unimodal. The proof of the following \cgp{result} is deferred to \Cref{apx:proofs:circ}.

\begin{prop}\label{prop:rob_is_unimod}
Let \(D\) be a dissimilarity matrix. The following are equivalent:
\begin{enumerate}
    \item \(D\) is circular Robinson (resp. strict circular Robinson).
    \item For any \(i\in [n]\) the function \(j\to D(i, i+j\bmod n)\) is unimodal (resp. strict unimodal).
\end{enumerate}
\end{prop}

This property is naturally invariant under cyclic permutations.

\begin{prop}
    A dissimilarity matrix $D$ is circular Robinson if and only if $\Pi_rD\Pi_r^T$ and $\Pi_s D\Pi_s^T$ are circular Robinson matrices.
\end{prop}

\begin{proof}
First, notice that the $(i,j)$ entry of $\Pi_sD\Pi_s^T$ and $\Pi_rD\Pi_r^T$ are $D(i+1\bmod n, j+1\bmod n)$ and $D(n-1-i,n- 1-j)$, respectively. Noticing that $\{D(i\bmod n, i+j\bmod n)\}_{j=0}^{n-1}$ is unimodal for all $i$, we have that $\{D(i+1\bmod n, i+j+1\bmod n)\}_{j=0}^{n-1}$ and $\{D(n-1-i\bmod n, n-1-i+j\bmod n)\}_{j=0}^{n-1}$ are unimodal. Therefore  $\Pi_rD\Pi_r^T$ and $\Pi_sD\Pi_s^T$ are also circular Robinson.
\end{proof}

Since \(\dihedral_n \cong \gen{\pi_r, \pi_s}\) it follows that circular Robinson matrices are invariant under the action of \(\dihedral_n\) by conjugation. This invariance is particular to the definition and should not be taken for granted. Other definitions proposed in the literature, e.g., \cite{recanati2018reconstructing}, do not enjoy this property. We believe that cyclic invariance makes the definition arguably more natural.

\subsection{Robinson orderings}
\label{sec:robinsonMatrices:robinsonOrderings}

The seriation problem does not assume we observe a linear or circular Robinson dissimilarity matrix, but instead its image under conjugation by an unknown permutation matrix. In other words, we observe matrices in $\preL$ and $\preC$. We call such matrices {\em Robinsonian matrices}.

Given a Robinsonian matrix and an algorithm for the corresponding seriation problem, the set of solutions may not be a singleton. In fact, the symmetries of linear and circular Robinson dissimilarity matrices \sac{{\em ensure}} they will never be a singleton. We call {\em Robinson orderings} \cgp{to} all the orderings  represented by the elements in the set of solutions.

Although there will never be a unique Robinson ordering, we can at least distinguish which ones are due to the natural symmetries of the problem. Therefore, for the linear seriation problem we call solutions in the same orbit under the action \(\dihedral_1\) the {\em trivial solutions} whereas those in different orbits {\em non-trivial solutions}. The same criteria applies for the circular seriation problem when the action of \(\dihedral_n\) is considered instead.

\section{The consecutive ones problem and \(PQ\)-trees}
\label{sec:consecutiveOnes}

Robinson matrices turn out to be natural to formulate the seriation problem. We now review the connection between this problem and the {\em consecutive ones problem}. This connection yields polynomial time algorithms for solving the seriation problem, and allows us to introduce \(PQ\)-trees, which will be extensively used in~\Cref{sec:algorithm}.

\subsection{The consecutive ones problem}
\label{sec:consecutiveOnes:consecutiveOnes}

The linear seriation problem is deeply connected to a combinatorial problem known as the {\em consecutive ones (C1) problem}. To introduce this problem, consider \cgp{an \(m\times n\)} binary matrix \(M\). The C1 problem is to find a permutation \(\Pi\) such that the entries of \(M\Pi\) equal to one appear consecutively along rows. We say \(M\) has the {\em consecutive ones (C1) property} if the C1 problem has a solution for \(M\). An example of such matrix can be found in \cref{ejem:pqtreeM}. The first linear time algorithm for the C1 problem was introduced by Booth and Lueker in~\cite{booth1975linear}. If $f$ is the number of ones in $M$ then their result states the C1 problem can be decided in \cgp{$\mathcal{O}(m+n+f)$} time.

%



An extension to this problem is the {\em circular ones (Cr1) problem}. The Cr1 problem is to find a permutation \(\Pi\) such that the entries of \(M\Pi\) equal to one appear consecutively {\em modulo \(n\)} along rows. We say \(M\) has the {\em circular ones (Cr1) property} if the Cr1 problem has a solution for \(M\)~ \cite{tucker1971matrix}. This problem can also be solved efficiently as it can be reduced to the C1 problem. Let \(\overline{M}\) be the matrix such that every row with a 1 on its first entry is complemented. Then \(M\) satisfies the C1 property if an only if \(\overline{M}\) satisfies the Cr1 property~\cite[Theorem 1]{tucker1971matrix}. Therefore, by forming the complement, the \(Cr1\) problem can be decided in polynomial time.

%


Both problems are connected to the seriation problem through the \cgp{ball} \sac{hypergraph} \(\HGr_D\), \cgp{introduced in \Cref{sec:seriation_dissim:rob_dissim}.} In fact, interval and arc hypergraphs are precisely those for which their incidence matrices respectively satisfy the C1 and Cr1 properties~\cite{kobler2017circular}. This suggests how to efficiently solve the seriation problem for Robinson matrices.

\begin{theorem}(\cite{chepoi1997recognition,mirkin1984graphs})
The linear and circular seriation problem can be reduced in polynomial time and space to deciding respectively the C1 and Cr1 problem. Robinson matrices can be recognized in $\mathcal{O}\left(n^{3}\right)$ time and with $\mathcal{O}\left(n^{3}\right)$ space.
\end{theorem}

The bounds above 
follow from the worst case in which \(\HGr_D\) has $\mathcal{O}\left(n^{2}\right)$ different hyperedges. In this case, for each of the $n$ possible centers and each row $i\in [n]$ the matrix can take $\mathcal{O}(n)$ possible values. In this case, the incidence matrix has $\mathcal{O}(n^3)$ entries.

\subsection{$PQ$-Trees}
\label{sec:consecutiveOnes:pqTrees}

The algorithmic structure underlying the algorithm to solve the C1 problem is the $PQ$-tree. A $PQ$-tree $\mathcal{T}$ on a set $\X$ is a rooted tree with two types of internal nodes denoted by $P$, represented as circles, and $Q$, represented as rectangles, and where the leaves represent the elements in $\X$. The type of node represents admissible permutations on \(\X\): children of a \(P\)-node can be permuted arbitrarily, whereas children of a \(Q\)-node can only be reversed.  \cref{ejem:pqtreeP} shows an example of a $PQ$-tree.



\(PQ\)-trees are related to the C1 problem as follows. Let \(Y_i\) be the indices of the columns of \(M\) such that its \(i\)-th entry equal to one. Then $\mathbf{Y}=\{Y_i\}_{i\geq 0}$ is a collection of subsets of \([n]\). The C1 problem can be solved if we can permute the elements of \([n]\) so that every \(Y_i\) becomes an interval. The algorithm starts with a single set \(\mathbf{Y}_1 = \{Y_{i_1}\}\) and determines the set of {\em admissible permutations} such that \(Y_{i_1}\) becomes an interval. These can be represented by a $PQ$-tree \(\mathcal{T}_1\) (see~\cite{booth1975linear} and~\cite{booth1976testing}). The algorithm proceeds by adding a \(Y_{i_2}\) to form \(\mathbf{Y}_2 = \{Y_{i_1},Y_{i_2}\}\) and update the \(PQ\)-tree accordingly. The main contribution of \cite{booth1976testing} is an algorithm for updating $\mathcal{T}_k$ in a way that given any subset $Y_k \subseteq [n]$, the set of permutations represented by the updated tree \(\mathcal{T}_{k+1}\) is precisely the set of admissible permutations of $\mathbf{Y}_{k+1}\cup \{Y_k\}$. This is done in time linear in the size of $Y_k$. The algorithm finishes when \(\mathbf{Y}\) is attained.

As an example, by considering all rows of the binary matrix in \cref{ejem:pqtreeM}, the resulting $PQ$-tree at the final step would be the one in \cref{ejem:pqtreeP}, and the solution set would be the one in \cref{ejem:pqtreeS}. 

%


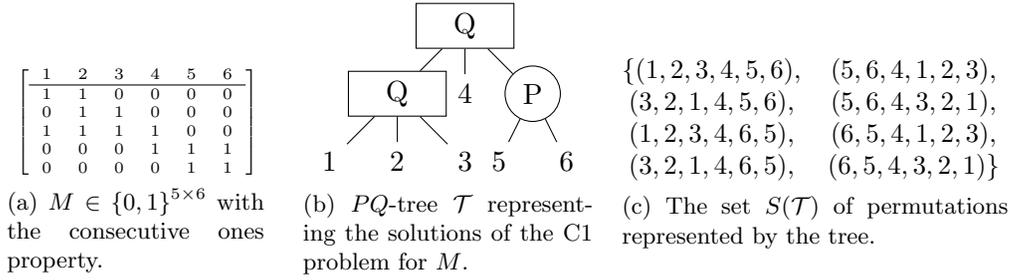
\begin{figure}
    \centering

    \subfloat[$M\in\{0,1\}^{5\times 6}$ with the consecutive ones property.]{\label{ejem:pqtreeM}
   \tiny $\left[\begin{array}{llllll}
1 & 2 & 3 & 4 & 5 & 6 \\
\hline
1 & 1 & 0 & 0 & 0 & 0 \\
0 & 1 & 1 & 0 & 0 & 0 \\
1 & 1 & 1 & 1 & 0 & 0 \\
0 & 0 & 0 & 1 & 1 & 1 \\
0 & 0 & 0 & 0 & 1 & 1
\end{array}\right]$
    }\quad
\adjustbox{valign=t, raise=1.65cm}{
    \subfloat[\sac{$PQ$-tree $\mathcal{T}$ representing the solutions of the \csl{C1 problem} for $M$.}]{\label{ejem:pqtreeP}
\begin{tikzpicture}[scale=0.6,
  level distance=1.5cm,
  level 1/.style={sibling distance=1.5cm},
  level 2/.style={sibling distance=1.5cm}]
  \node[rectangle,minimum width=1.3cm, draw] {Q}
    child {node[rectangle,minimum width=1.3cm, draw] {Q}
        child {node {1}}
        child {node {2}}
        child {node {3}}
        }
    child {node {4}}
    child {node[circle, draw] {P}
      child {node {5}}
      child {node {6}}
        };
\end{tikzpicture}
}}%
\quad
    \subfloat[The set $S(\mathcal{T})$ of permutations represented by the tree.]{\label{ejem:pqtreeS}\small
        $\begin{matrix}
     \{(1,2,3,4,5,6), & (5,6, 4, 1,2,3),\\
     (3,2,1,4,5,6), & (5,6, 4, 3,2,1),\\
     (1,2,3,4,6,5), & (6,5, 4, 1,2,3),\\
     (3,2,1,4,6,5), & (6,5, 4, 3,2,1)\}
    \end{matrix}$
    }%

    \caption{A $PQ$-tree of all solutions to the \csl{C1 problem} for a $\{0,1\}$-matrix.
    }
    \label{ejem:pqtree}
\end{figure}

\section{Optimal algorithm for strict circular seriation}\label{sec:algorithm}

In this section, we present an optimal algorithm for circular seriation in the strict Robinson case. \csl{We note in passing that our algorithm works as well for the strict linear case, but we omit this variant}. Our algorithm runs in $\mathcal{O}(n^2)$ time and space, which is obviously optimal, since it is the time required to read the input and the space required to provide a strict Robinson dissimilarity.\footnote{Given the strict Robinson property, it is clear that the underlying matrix is dense, and therefore $\mathcal{O}(n^2)$ memory is required to even provide the input.}
The core algorithm relies \csl{on} two main ideas: merging nearest neighbors, and discarding forbidden arc reversals. We recursively merge nearest neighbors, using the fact that nearest neighbors are guaranteed to be consecutive elements in a strict Robinsonian ordering. Exploiting this fact we can obtain chains of consecutive elements which are stored in $Q$-nodes of $PQ$-trees. The most \csl{delicate} part of the algorithm consists in \cgp{efficiently deciding whether}  
each $Q$-node can be \textit{uniquely oriented}. In that case, such $Q$-node can be deleted and its children merged to the parent $Q$-node. \sac{We will indistinctly refer to this operation as the $Q$-node being \textit{fixed} or \textit{oriented}}. The process of building chains of consecutive elements and deciding their orientation can be done in several ways. The advantage of our algorithm is that the total number \csl{of} comparisons to decide each orientation is bounded by $\mathcal{O}(n)$, which leads to a running time ${\cal O}(n^2)$.

\csl{We provide preliminary results in~\Cref{sec:algo:prelimI} and~\Cref{sec:algo:prelimII} about the arc structure of the nearest-neighbour graph, and the treatment of arc reversals. Then, in~\Cref{sec:algo:main} we develop our algorithm, whose correctness and optimality are proved in \Cref{subsec:Analysis_Rec_Seriation}.}



\subsection{Preliminaries part I: nearest neighbours graph in strict Robinson dissimilarities}
\label{sec:algo:prelimI}


\cgp{One of our algorithmic building blocks is based on the idea that in the case of {\em strict dissimilarities}, a pair of nearest neighbours must lie consecutively in {\em any} Robinson ordering. In this subsection, we prove this fact and use it to justify the first step of our algorithm, which is based on merging nearest neighbors to reduce the size of the instance.} 

Given $D\in \preC$, a collection of subsets $\mathcal{P}=\{\mathcal{I}_i\}_{i\geq 0}$ of $\X$ is said to be an arc partition if $\mathcal{P}$ is a partition of $\X$, 
and every set $\mathcal{I}_i$ is an arc of consecutive elements in any Robinson ordering. 
The set of nearest neighbours of $x\in \X$ is defined as $\operatorname{NN}(x)\defined  \arg\min_{y\in \X\setminus \{x\}} \mathbf{d}(x,y)$. The nearest-neighbours graph is an undirected graph $\NNG(\X,\mathbf{d}) =(\X,\mathcal{E})$ such that $\{x,y\} \in \mathcal{E}$ iff $x\in \NN(y)$ or $y\in \NN(x)$. 
An essential condition of strict Robinson dissimilarities is what we called the \textit{nearest-neighbour condition}, which implies that the connected components of the nearest-neighbours graph correspond to arcs of consecutive elements, and since connected components form a partition, such collection corresponds to an arc partition.


 \begin{defi}{(Nearest-neighbour condition)}
 A dissimilarity matrix $D\in\mathbb{R}^{n\times n}$ is said to have the \textit{nearest-neighbour condition} if it holds that $\NN(i) \subseteq V_i^C\defined \{i-1\bmod n, i+1\bmod n\}$.\footnote{Given some enumeration $\enu$ and $i\in [n]$, when we write $\NN(i)$, we refer to the set $\enu(\NN(x_i))$.}
\end{defi}


It is immediate to verify that {\em strict} circular Robinson dissimilarities satisfy the nearest-neighbour condition, which is not necessarily true in the non-strict case.







We also recall from graph theory that given a graph $G=(\X,\mathcal{E})$ and a node $x\in \X$, the set of adjacent nodes to $x$ is denoted as $\mathcal{N}_G(x)\defined  \{y\in \X: \{x,y\}\in \mathcal{E} \}$. The function $x\mapsto \mathcal{N}_G(x)$ is called the neighbourhood.
The \sac{cycle graph $\cycle{n}=([n], \mathcal{E})$ is} the graph with edge set $\mathcal{E}=\{\{i,(i+1)\mod n\}: i\in [n]\}$. If a graph $G$ is a subgraph of $\cycle{n}$, then it is clear that its connected components correspond to arcs of $([n], \mathscr{C}_n)$. A direct consequence of the nearest-neighbour condition is that the nearest-neighbours \sac{graphs} of \textit{strict} circular Robinson dissimilarity matrices correspond to subgraphs of $\cycle{n}$.

Our algorithm relies crucially on the fact that strict dissimilarities must respect nearest neighbors in any Robinson ordering. This is not necessarily true in the non-strict case.
\begin{lemma}\label{lema:consecutive}
Let $D\in \presC$ and $i\in [n]$. Suppose that $j\in \NN(i)$, then in any Robinson ordering $\sigma$, the elements $i$ and $j$ are consecutive.
\end{lemma}

\begin{proof}
Let $\sigma$ be any Robinson ordering. Let $j\in \NN(i)$ and let $r\defined  D(i,j)$. This implies that for any $k\in B(i,r)\setminus \{i\}$, $D(i,k)=D(i,j)$. Suppose by contradiction that there exist $k_1,i,k_2$ consecutive in $\sigma$, with $j\neq k_1,k_2$. Since in any Robinson ordering balls are arcs, this implies that either $k_1\in B(i,r)$ or $k_2\in B(i,r)$. Any of the two cases is a contradiction with the nearest-neighbour condition, proving the result.
\end{proof}

Since nearest-neighbours must be consecutive, we get that the connected components of the nearest-neighbours graph of a \textit{strict} circular Robinson dissimilarity correspond to arcs of any Robinson ordering. Hence, the set of connected components constitute an arc partition. The fact that this graph is a subgraph of the cycle graph makes computationally efficient finding the order intrinsic to each component, and the task is divided in two steps:
\begin{enumerate}[leftmargin=24pt]
    \item Find all degree $1$ nodes. These correspond to the borders of the components.
    \item Perform Depth-First Search \csl{(\DFS)} (\cref{alg:dfs}) starting at each non visited degree $1$ node. The order \cgp{of visits} 
    will follow the Robinson ordering (or backwards).
\end{enumerate}
If there are no degree one nodes,  \sac{then\footnote{We use the symbol $\cong$ to either denote graph and group isomorphism} $\NNG \cong \cycle{n}$} and therefore we can start at any node. For an algorithmic implementation, tuples can be used to represent the local fragments of Robinson orderings ($Q$-nodes). A tuple is an ordered set $\alpha=(a_{0},a_{1},a_{2},\ldots ,a_{k-1})$. We write $\alpha(i)$ to denote $a_i$, the $i$-th element of $\alpha$. Each connected component will be stored in a tuple $\alpha$, where $\alpha(j)$ is the $j$-th element visited by performing a \DFS. The procedure is summarized in the procedure \partition (\cref{alg:gac}), whose correctness is stated in the following Proposition (the proof of the next result is omitted for brevity).

\begin{prop}\label{prop:dfs}
Given \cgp{a matrix} 
$D\in \presC$, by performing \partition\, (\cref{alg:gac}) with input $([n],D)$ the resulting tuples follow an arc ordering for every Robinson ordering. 
\end{prop}





\subsection{Preliminaries part II: orienting arcs}\label{sec:orienting}
\label{sec:algo:prelimII}

The previous section tell us that nearest-neighbours must be consecutive in the strict Robinson case. By exploiting this idea we can obtain ordered sequences of elements stored in $Q$-nodes of a $PQ$-tree. Notice however that $Q$-nodes are allowed to be reversed, which at this point of the algorithm is not guaranteed to lead to Robinson orderings.
If this is not the case, the inconsistent ordering must be discarded, which corresponds to removing the $Q$-node \sac{and merging its children} directly to the parent $Q$-node. We call this process \cgp{\em orientation}. 
In this section we provide computationally efficient conditions to determine when a $Q$-node must be oriented. Each $Q$-node $\alpha$ in a tree $\mathcal{T}$ can be associated with an arc $\mathcal{I}_\alpha$ in $\X$: the arc of all leaves in $\X$ which are descendants of $\alpha$. Reversing $\alpha$ corresponds to reversing $\mathcal{I}_\alpha$. The first relevant concept to determine when it is possible to reverse each arc is the strictly overlapping condition, which has been studied  for instance in \cite{quilliot1984circular} and in \cite{kobler2017circular}. An example of the property can be seen in \cref{fig:orient}(a).


\begin{defi}\label{defi:overlap}
Two arcs $\mathcal{I}$ and $\mathcal{J}$ are said to strictly overlap, denoted by $\mathcal{I} \between^* \mathcal{J}$, if
\[ 
    1.\,\, \mathcal{I}	\not\subset \mathcal{J};
    \qquad 2.\,\, \mathcal{J}	\not\subset \mathcal{I};
    \qquad 3.\,\, \mathcal{I}^c	\not\subset \mathcal{J};
    \quad \mbox{and} \quad 4.\,\, \mathcal{J}	\not\subset \mathcal{I}^c.
\]
\end{defi}

\begin{obs}\label{inter}
The relation $\between^*$ is symmetric and equivalent to
\[
1.\, \mathcal{I}		\cap  \mathcal{J}^c \neq \emptyset;
\qquad   2.\, \mathcal{J}	\cap  \mathcal{I}^c \neq \emptyset;
\qquad 3.\, \mathcal{I}^c	\cap \mathcal{J}^c \neq \emptyset;
\quad\mbox{and}\quad 4.\, \mathcal{I}	\cap \mathcal{J} \neq \emptyset.
\]
\end{obs}

\begin{lemma}\label{lema:borders}
Let $\mathcal{I}$ and $\mathcal{J}$ be two arcs. Let $a,b$ and $a^\prime,b^\prime$ be the borders of $\mathcal{I}$ and $\mathcal{I}^{c}$, respectively, where $a$ and $a^\prime$ \sac{($b$ and $b^\prime$)} are consecutive in the cyclic order.
Then $\mathcal{I} \between^* \mathcal{J}$ if and only if one of the following conditions holds
$$
(i)\, \{a,a^\prime\}\subset \mathcal{J} \textit{ and } \{b,b^\prime\}\subset \mathcal{J}^c;
\qquad \mbox{or} \qquad   (ii)\, \{b,b^\prime\}\subset \mathcal{J} \textit{ and } \{a,a^\prime\}\subset \mathcal{J}^c.
$$

\end{lemma}
\begin{proof} We first prove ($\Leftarrow$).  
        Suppose (i) holds (the other case follows analogously). Then since $a\in \mathcal{I}$ and $a^\prime \in \mathcal{I}^c$ we get conditions 2 and 4 of \cref{inter}. Now, since $b\in \mathcal{I}$ and $b^\prime \in \mathcal{I}^c$ we get conditions 1 and 3 of  \cref{inter}.

        Next we prove ($\Rightarrow$). First we notice that there are at least two elements in $\mathcal{I}$ and two elements in $\mathcal{I}^c$ (otherwise containing a single element of these arcs would imply containing the whole set, contradicting one of the conditions in \cref{defi:overlap}). Hence, the elements $a,b,a^\prime$ and $b^\prime$ exist and are distinct. 
            Suppose $a\in \mathcal{J}$ (the case  $a\in \mathcal{J}^c$ is analogous), and let $z\in \mathcal{J}\setminus \mathcal{I}$ (exists by hypothesis).
            Since $\mathcal{J}$ is an arc, it must contain one of the two paths connecting $a$ and $z$. Since it does not contain the whole $\mathcal{I}$ it must be the path that covers $a^\prime$, therefore $\{a,a^\prime\}\subset \mathcal{J}$ and $b\notin \mathcal{J}$. On the other hand, since it does not contain the whole $\mathcal{I}^c$, \sac{$b^\prime \notin \mathcal{J}$}. It follows that $\{b,b^\prime\}\subset \mathcal{J}^c$.
\end{proof}

Given an arc $\mathcal{I}=\{a_0,\ldots,a_{k-1}\}$ (where elements are indexed following the cyclic order), we define the {\em permutation that reverses ${\cal I}$} as
the permutation $\sigma$ s.t.~$\sigma(a_j)=a_{k-j-1}$ for $j\in\{0,\ldots,k-1\}$, and $\sigma(x)=x$ if $x\notin {\cal I}$.

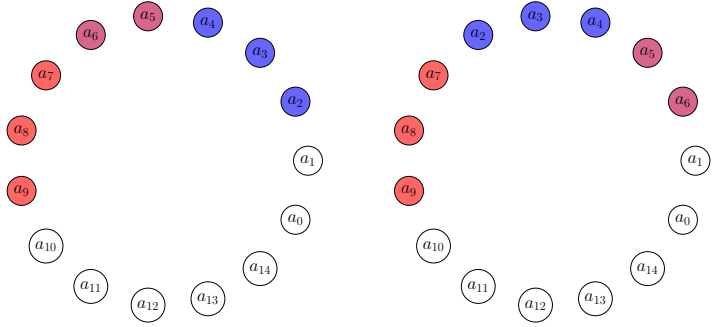
\begin{figure}
    \centering

\scalebox{0.385}{

    \begin{tikzpicture}[thick]
 
\def \n {15}
\def \radius {5cm}
\def \margin {7} 
\node[draw, circle, minimum size=6pt, thick,text opacity = 1,  fill opacity=0.6, fill=white ] at ({360/\n * (0- 1)}:\radius) {\textcolor{black}{\LARGE $a_{0}$}};
\node[draw, circle, minimum size=6pt, thick,text opacity = 1,  fill opacity=0.6, fill=white ] at ({360/\n * (1- 1)}:\radius) {\textcolor{black}{\LARGE $a_{1}$}};
\node[draw, circle, minimum size=6pt, thick,text opacity = 1,  fill opacity=0.6, fill=blue ] at ({360/\n * (2- 1)}:\radius) {\textcolor{black}{\LARGE $a_{2}$}};
\node[draw, circle, minimum size=6pt, thick,text opacity = 1,  fill opacity=0.6, fill=blue ] at ({360/\n * (3- 1)}:\radius) {\textcolor{black}{\LARGE $a_{3}$}};
\node[draw, circle, minimum size=6pt, thick,text opacity = 1,  fill opacity=0.6, fill=blue ] at ({360/\n * (4- 1)}:\radius) {\textcolor{black}{\LARGE $a_{4}$}};
\node[draw, circle, minimum size=6pt, thick,text opacity = 1,  fill opacity=0.6, fill=purple ] at ({360/\n * (5- 1)}:\radius) {\textcolor{black}{\LARGE $a_{5}$}};
\node[draw, circle, minimum size=6pt, thick,text opacity = 1,  fill opacity=0.6, fill=purple ] at ({360/\n * (6- 1)}:\radius) {\textcolor{black}{\LARGE $a_{6}$}};
\node[draw, circle, minimum size=6pt, thick,text opacity = 1,  fill opacity=0.6, fill=red ] at ({360/\n * (7- 1)}:\radius) {\textcolor{black}{\LARGE $a_{7}$}};
\node[draw, circle, minimum size=6pt, thick,text opacity = 1,  fill opacity=0.6, fill=red ] at ({360/\n * (8- 1)}:\radius) {\textcolor{black}{\LARGE $a_{8}$}};
\node[draw, circle, minimum size=6pt, thick,text opacity = 1,  fill opacity=0.6, fill=red ] at ({360/\n * (9- 1)}:\radius) {\textcolor{black}{\LARGE $a_{9}$}};
\node[draw, circle, minimum size=6pt, thick,text opacity = 1,  fill opacity=0.6, fill=white ] at ({360/\n * (10- 1)}:\radius) {\textcolor{black}{\LARGE $a_{10}$}};
\node[draw, circle, minimum size=6pt, thick,text opacity = 1,  fill opacity=0.6, fill=white ] at ({360/\n * (11- 1)}:\radius) {\textcolor{black}{\LARGE $a_{11}$}};
\node[draw, circle, minimum size=6pt, thick,text opacity = 1,  fill opacity=0.6, fill=white ] at ({360/\n * (12- 1)}:\radius) {\textcolor{black}{\LARGE $a_{12}$}};
\node[draw, circle, minimum size=6pt, thick,text opacity = 1,  fill opacity=0.6, fill=white ] at ({360/\n * (13- 1)}:\radius) {\textcolor{black}{\LARGE $a_{13}$}};
\node[draw, circle, minimum size=6pt, thick,text opacity = 1,  fill opacity=0.6, fill=white ] at ({360/\n * (14- 1)}:\radius) {\textcolor{black}{\LARGE $a_{14}$}};
\end{tikzpicture}}
\qquad
\scalebox{0.385}{

    \begin{tikzpicture}[thick]
 
\def \n {15}
\def \radius {5cm}
\def \margin {7} 

\node[draw, circle, minimum size=6pt, thick,text opacity = 1,  fill opacity=0.6, fill=white ] at ({360/\n * (0- 1)}:\radius) {\textcolor{black}{\LARGE $a_{0}$}};
\node[draw, circle, minimum size=6pt, thick,text opacity = 1,  fill opacity=0.6, fill=white ] at ({360/\n * (1- 1)}:\radius) {\textcolor{black}{\LARGE $a_{1}$}};
\node[draw, circle, minimum size=6pt, thick,text opacity = 1,  fill opacity=0.6, fill=purple ] at ({360/\n * (2- 1)}:\radius) {\textcolor{black}{\LARGE $a_{6}$}};
\node[draw, circle, minimum size=6pt, thick,text opacity = 1,  fill opacity=0.6, fill=purple ] at ({360/\n * (3- 1)}:\radius) {\textcolor{black}{\LARGE $a_{5}$}};
\node[draw, circle, minimum size=6pt, thick,text opacity = 1,  fill opacity=0.6, fill=blue ] at ({360/\n * (4- 1)}:\radius) {\textcolor{black}{\LARGE $a_{4}$}};
\node[draw, circle, minimum size=6pt, thick,text opacity = 1,  fill opacity=0.6, fill=blue ] at ({360/\n * (5- 1)}:\radius) {\textcolor{black}{\LARGE $a_{3}$}};
\node[draw, circle, minimum size=6pt, thick,text opacity = 1,  fill opacity=0.6, fill=blue ] at ({360/\n * (6- 1)}:\radius) {\textcolor{black}{\LARGE $a_{2}$}};
\node[draw, circle, minimum size=6pt, thick,text opacity = 1,  fill opacity=0.6, fill=red ] at ({360/\n * (7- 1)}:\radius) {\textcolor{black}{\LARGE $a_{7}$}};
\node[draw, circle, minimum size=6pt, thick,text opacity = 1,  fill opacity=0.6, fill=red ] at ({360/\n * (8- 1)}:\radius) {\textcolor{black}{\LARGE $a_{8}$}};
\node[draw, circle, minimum size=6pt, thick,text opacity = 1,  fill opacity=0.6, fill=red ] at ({360/\n * (9- 1)}:\radius) {\textcolor{black}{\LARGE $a_{9}$}};
\node[draw, circle, minimum size=6pt, thick,text opacity = 1,  fill opacity=0.6, fill=white ] at ({360/\n * (10- 1)}:\radius) {\textcolor{black}{\LARGE $a_{10}$}};
\node[draw, circle, minimum size=6pt, thick,text opacity = 1,  fill opacity=0.6, fill=white ] at ({360/\n * (11- 1)}:\radius) {\textcolor{black}{\LARGE $a_{11}$}};
\node[draw, circle, minimum size=6pt, thick,text opacity = 1,  fill opacity=0.6, fill=white ] at ({360/\n * (12- 1)}:\radius) {\textcolor{black}{\LARGE $a_{12}$}};
\node[draw, circle, minimum size=6pt, thick,text opacity = 1,  fill opacity=0.6, fill=white ] at ({360/\n * (13- 1)}:\radius) {\textcolor{black}{\LARGE $a_{13}$}};
\node[draw, circle, minimum size=6pt, thick,text opacity = 1,  fill opacity=0.6, fill=white ] at ({360/\n * (14- 1)}:\radius) {\textcolor{black}{\LARGE $a_{14}$}};
\end{tikzpicture}}

\caption{
\cgp{In the left, two strictly overlapping arcs $\mathcal{I} = \{a_i\}_{i=2}^6$ (blue and purple) and $\mathcal{J}= \{a_i\}_{i=5}^9$ (red and purple). The intersection $\mathcal{I}\cap \mathcal{J} = \{a_5, a_6\}$ in purple. In the right, we have the ordering after the action of the permutation $\sigma$ that reverses the elements of ${\cal I}$.}}
\label{fig:orient}%
\end{figure}

\begin{lemma}\label{lema:canreverse}
Let $\mathcal{I}$ and $\mathcal{J}$ be two arcs and let $\sigma$ be the permutation that reverses the elements of $\mathcal{I}$. Then $\mathcal{I}\between^*\mathcal{J}$ iff the permutation of ${\cal J}$ by $\sigma$ is not an arc.

\end{lemma}
\begin{proof} We first prove ($\Leftarrow$). By contraposition, assume any of the conditions in \cref{defi:overlap} do not hold, then it is easy to see that $\sigma(\mathcal{J})=\mathcal{J}$, which is an arc. 
Now we prove ($\Rightarrow$). If $\mathcal{I}\between^*\mathcal{J}$, then at least one of the conditions of \cref{lema:borders} hold. 
Since $\sigma(a)=b$, $\sigma(b)=a$, $\sigma(a^\prime)=a^\prime$ and $\sigma(b^\prime)=b^\prime$ then $\sigma(\mathcal{J})$ is not connected and thus it is not an arc.
\end{proof}

As an example, consider the two strictly overlapping arcs $\mathcal{I}$ and $\mathcal{J}$ in \cref{fig:orient}. \cgp{By reversing ${\cal I}$, the arc ${\cal J}$ gets ripped appart into two disconnected pieces: namely its red nodes and its purple nodes.} 
Recall from \cref{prop:arc_is_ball} that a dissimilarity matrix is circular Robinson iff each ball $\mathcal{J}$ is an arc. Therefore, any arc $\mathcal{I}$ cannot be arbitrarily reversed \sac{to produce} a new Robinson ordering iff there is some ball that strictly overlaps with $\mathcal{I}$. In terms of $PQ$-trees, a necessary and sufficient condition for a $Q$-node $\alpha$ to be orientable is the existence of some $z\in \X$ and $r>0$ such that the ball $B_r(z)$ strictly overlaps with $\mathcal{I}_\alpha$. In such case, one of the two orientations of the node is not compatible with a Robinson ordering since in one of these orientations the ball gets disconnected. By \cref{lema:borders}, to determine the orientation of the arc $\mathcal{I}$ one could equivalently check whether there exists $z\in {\cal X}$ and $r>0$ such that
\[
\big[ \{a,a^\prime\}\subset B_r(z) \,\,\wedge\,\, \{b,b^\prime\}\subset B_r(z)^c \big]
\quad \vee \quad
\big[\{b,b^\prime\}\subset B_r(z) \,\,\wedge\,\, \{a,a^\prime\}\subset B_r(z)^c \big],
\]
If none of this conditions hold, then we say the arc is \textit{not orientable} which means that the $Q$-node in the tree must be preserved. Notice that this requires knowing that $a,b$ ($a^\prime,b^\prime$) are the borders of $\mathcal{I}$ (resp. $\mathcal{I}^c$) in advance. \brute (\cref{alg:bruteforce}) is an efficient way for orienting the arc $\mathcal{I}$ with respect to the dissimilarity $\mathbf{d}$ when we have \textit{border candidates} but the actual borders within the candidates are unknown. 


\begin{defi}[Border candidates \textit{of an arc}]\label{defi:candidates}
An 4-tuple of sets $(\mathcal{A}^\prime,\mathcal{A},\mathcal{B},\mathcal{B}^\prime)$ are said to be border candidates of the arc $\mathcal{I}$, if the following properties hold:
\begin{enumerate}
    \item $\mathcal{A},\mathcal{B}\subset \mathcal{I}$ and $\mathcal{A}^\prime ,\mathcal{B}^\prime \subset \mathcal{I}^c$.
    \item The sets are pairwise disjoint.
    \item If $a,b$ are the borders of $\mathcal{I}$ and $a^\prime,b^\prime$ are the borders of $\mathcal{I}^c$, then: $a\in \mathcal{A}, b\in \mathcal{B},a^\prime\in \mathcal{A}^\prime$ and $b^\prime\in \mathcal{B}^\prime$.
    \item Either $(\mathcal{A}^\prime,\mathcal{A},\mathcal{B},\mathcal{B}^\prime)$ or $(\mathcal{A}^\prime,\mathcal{B},\mathcal{A}, \mathcal{B}^\prime)$ is cyclically ordered.\footnote{Formally, the ordered collection is a consistent cyclic quasi order, see \cref{defi:quasiorder}}
\end{enumerate}
\end{defi}

The next result provides correctness for \brute (\cref{alg:bruteforce}) (see its proof in \Cref{app:bruteiscorrect}).

\begin{lemma}\label{lema:bruteiscorrect}
Let $\mathcal{I}\subset \X$ be an arc in any Robinson ordering. Suppose $a,b$ are the borders of $\mathcal{I}$ and $a^\prime,b^\prime$ are the respective borders of $\mathcal{I}^c$. Additionally suppose that $(\mathcal{A}^\prime,\mathcal{A},\mathcal{B},\mathcal{B}^\prime)$ are border candidates for $\mathcal{I}$. Then, for every $z\in \X$, the following statements are equivalent:
\begin{enumerate}
    \item Both $\{a,a^\prime\}\subset B_r(z)$ and $\{b,b^\prime\}\subset B_r(z)^c$ hold \sac{for some $r>0$.}
    \item There exist $(x^\prime,x,y,y^\prime) \in \mathcal{A}^\prime\times \mathcal{A}\times\mathcal{B}\times\mathcal{B}^\prime$ s.t.~$\max \{ f_z(x), f_z(x^{\prime})\} < \min \{ f_z(y), f_z(y^{\prime})\}$.
        \item It holds that $\max\big\{\min f_z(\mathcal{A}),\min f_z(\mathcal{A}^\prime )\big\}
    < \min\big\{\max f_z(\mathcal{B}),\max f_z(\mathcal{B}^\prime ) \big\}$.
\end{enumerate}
Above $f_z(\cdot) \defined \mathbf{d}(z,\cdot)$ and given any $U\subset \X$, \sac{$\min f_z(U) \defined \min_{y\in U}f_z(y)$ (similar for $\max$)}.
\end{lemma}

\begin{coro}\label{coro:brutecorrect}
Let $\mathcal{I}\subset \X$ be an arc in any Robinson ordering and suppose the sets $(\mathcal{A}^\prime,\mathcal{A},\mathcal{B},\mathcal{B}^\prime)$ are \textit{border candidates} for the arc $\mathcal{I}$. Then, \cref{alg:bruteforce} correctly determines if $\mathcal{I}$ must be fixed, reversed or if it is not orientable.
\end{coro}

\begin{algorithm}
 \caption{\brute}\label{alg:bruteforce}
\begin{algorithmic}[1]
\small
\STATE{\textbf{Input}: A sequence of sets $(\mathcal{A}^\prime,\mathcal{A},\mathcal{B},\mathcal{B}^\prime)$}

\STATE{Let $f_z(x)\defined  \mathbf{d}(z,x)$ for every $z\in \X$}
\STATE{Let $O_i: \X \rightarrow \{True,False\}$ for $i={1,2,3,4}$ be defined by}
\sac{
\STATE{$O_1(z)\defined 
\max \big\{  \min f_z(\mathcal{A}),\min f_z(\mathcal{A}^\prime ) \big\}
< \min  \big\{ \max f_z(\mathcal{B}),\max f_z(\mathcal{B}^\prime )\big\}$}
\STATE{$O_2(z)\defined 
\max \big\{ \min f_z(\mathcal{B}),\min f_z(\mathcal{B}^\prime )\big\}
< \min \big\{ \max f_z(\mathcal{A}),\max f_z(\mathcal{A}^\prime ) \big\}$}
\STATE{$O_3(z)\defined 
\max \big\{ \min f_z(\mathcal{A}),\min f_z(\mathcal{B}^\prime )\big\}
< \min \big\{ \max f_z(\mathcal{A}^\prime),\max f_z(\mathcal{B} ) \big\}$}
\STATE{$O_4(z)\defined 
\max \big\{ \min f_z(\mathcal{B}),\min f_z(\mathcal{A}^\prime )\big\}
< \min \big\{ \max f_z(\mathcal{B}^\prime),\max f_z(\mathcal{A} ) \big\}$}}
\FOR{$z\in \X$}
    \IF{$O_1(z)\vee O_2(z)$}
        \RETURN{\texttt{`correct'}}
    \ELSIF{$O_3(z)\vee O_4(z)$}
         \RETURN{\texttt{`reverse'}}
    \ENDIF
\ENDFOR
\RETURN{\texttt{`not orientable'}}
\STATE{\textbf{Output}: A string determining the orientation of the input}
\end{algorithmic}
\end{algorithm}%

\begin{obs} \label{obs:complexitybrute}
The time complexity of \cref{alg:bruteforce} with input $(\mathcal{A}^\prime,\mathcal{A},\mathcal{B},\mathcal{B}^\prime)$ is $\mathcal{O}(|\mathcal{X}|\cdot\max\{|\mathcal{A}^\prime|,|\mathcal{A}|,|\mathcal{B}|,|\mathcal{B}^\prime|\} ),$
thus it is an efficient way of orienting a $Q$-node $\alpha$ whenever the sets of border candidates for $\mathcal{I}_\alpha$ is not too big.
\end{obs}

\subsection{The recursive seriation algorithm}
\label{sec:algo:main}



\cgp{
We devote the next subsection to describing our main algorithm: \main (\Cref{alg:Rec_seriation}). This algorithm starts from the singleton elements of ${\cal X}$, and proceeds by recursively detecting nearest neighbors among previously computed arcs, merging them, and deciding whether the merged sets must be oriented. An important ingredient of our algorithm is the use of $Q$-trees: this is a useful data structure to merge arcs, as well as fixing the orientations when they are detected.
In order to run in total quadratic time, great care is needed regarding the number of comparisons used to decide the orientations, for which we find useful to compute {\em border candidates} for $Q$-trees.
}


\cgp{For the sake of readibility, we dissect the main algorithm in terms of various subroutines, which are presented and analyzed separately. Hence, we structure our presentation as follows. In \Cref{sssec:Initialization} we specify the initialization of \main. In \Cref{sssec:Border_Candidates}, we describe the key operation of computing the border candidates of a $Q$-tree, summarized in the procedure \bordercandidates (\Cref{alg:bc}). Next, in \Cref{sssec:Tree_Dissimilarity} we describe the operation of computing dissimilarities among $Q$-trees, and how this leads to the recursion behind \main. In \Cref{sssec:Consecutive_Orientation} we study the process of deciding the orientation among consecutive children of a given $Q$-node, including a pseudocode of this procedure \orient (\Cref{alg:orient}). Next, in \Cref{sssec:Complete_Internal_Orientation} we use the previous subroutine and complement it with the additional steps required to completely orient the internal nodes of a $Q$-tree; the associated subroutines here are in the procedures \complete (\Cref{alg:complete}), and \completefinal (\Cref{alg:completefinal}). 
In \Cref{sssec:External_Orientation} we provide a method to decide the orientation between trees that have been detected to be nearest neighbours; the useful subroutine here is the \external (\Cref{alg:external}).
}

\begin{algorithm}
\small
 \caption{\main}
\begin{algorithmic}[1]
\label{alg:Rec_seriation}
\STATE{\textbf{Input}: $({\X},\mathbf{d}, \mathbf{T})$: $\mathbf{d}$ dissimilarity over ${\X}$; $\mathbf{T}$ family of $Q$-trees with leaves given by elements of $\X$}
\FOR{${\cal T}\in \mathbf{T}$}
\FOR{${\cal T}^\prime \in \mathbf{T}$}
\STATE{$\dmin(\mathcal{T},\mathcal{T}^\prime)$,  $\argdmin(\mathcal{T},\mathcal{T}^\prime)= \dissimilarity (\mathcal{T},\mathcal{T}^\prime, \mathbf{d})$} %
\STATE{$\external({\cal T},{\cal T}^\prime, \argdmin (\mathcal{T},\mathcal{T}^\prime))$}
\ENDFOR
\ENDFOR
\STATE{ 
$\mathbf{T}^\prime = \partition(\dmin, \mathbf{T})$}
\IF{ $|\mathbf{T}^\prime|=1$}
\STATE{ $\completefinal (\mathcal{T})$} \quad \COMMENT{where $\mathbf{T}^\prime=\{\mathcal{T}\}$}
\RETURN{ $\mathcal{T}$ }
\ELSE
\FOR{ $\mathcal{T}\in \mathbf{T}^\prime$}
\STATE{
\complete$(\mathcal{T})$}
\ENDFOR
\RETURN{\main$({\X},\mathbf{d},\mathbf{T}^{\prime})$}
\ENDIF
\STATE{\textbf{Output:} A $Q$-tree ${\cal T}$ containing all Robinson orderings}
\end{algorithmic}
\end{algorithm}

\subsubsection{Initialization} \label{sssec:Initialization}

\sac{\main receives as input a family $\mathbf{T}$ of $Q$-trees (which are $PQ$-trees composed solely by $Q$-nodes), with leaves corresponding to elements of $\X$. Given an instance $(\X,\mathbf{d})$, we initialize  the algorithm with $(\X,\mathbf{d},\mathbf{T})$, where $\mathbf{T}=\{x: x\in \X\}$.} 
For each $\mathcal{T}\in \mathbf{T}$ we write as $\leaves{\mathcal{T}}$ the set of leaves of $\mathcal{T}$. Hence, in the initial case we get $\leaves{x}\defined \{x\}$.


\subsubsection{Computing border candidates}
\label{sssec:Border_Candidates}


We endow each $\mathcal{T}\in \mathbf{T}$ with a set $\borders{\mathcal{T}}$ of \textit{border candidates}\footnote{We emphasize the distinction of the border candidates \textit{of a tree}, \cgp{which we are about to introduce}; and the border candidates \textit{of an arc}, introduced in \cref{defi:candidates}}, which are all leaves of $\mathcal{T}$ that appear in the extreme left or right under some configuration of the tree. Whenever $|\leaves{\mathcal{T}}|\geq 2$, \sac{and conditionally on an given orientation of the root,} the set of border candidates $\borders{\mathcal{T}}$ can be split in two: left and right. The set of left border candidates, denoted as $\lborders{\mathcal{T}}$, are all elements in $\borders{\mathcal{T}}$ that appear in the extreme left under some configuration of the tree.
\cgp{The definition of $\rborders{\mathcal{T}}$ is entirely analogous with appearing on the right.} For instance, in the tree $\mathcal{T}$ appearing in \cref{fig:nspoon}, we have $\lborders{\mathcal{T}}=\{a_3,b_3,b_2,b_1\}$ and $\rborders{\mathcal{T}}=\{b_0\}$.


Let $\alpha$ denote the root of $\mathcal{T}$. Whenever $\depth(\mathcal{T})>1$, where $\depth(\mathcal{T})$ denotes the tree-depth of $\mathcal{T}$, 
$\mathcal{T}^L$ (resp. $\mathcal{T}^R$) denotes the subtree of $\mathcal{T}$ whose root is the first (resp. last) $Q$-node among the direct descendants of $\alpha$. \sac{Notice that for the computation of the border candidates of $\mathcal{T}$ it is convenient to consider the relations: $\lborders{\mathcal{T}} = \borders{\mathcal{T}^L}$ and $\rborders{\mathcal{T}} =\borders{\mathcal{T}^R}$. Thus, to obtain $\borders{\mathcal{T}}  = \borders{\mathcal{T}^L} \cup \borders{\mathcal{T}^R} $ it suffices to recursively call for the border candidates of the subtrees $\mathcal{T}^L$ and $\mathcal{T}^R$. \bordercandidates (\cref{alg:bc}) is an straightforward implementation of this idea that runs in time $\mathcal{O}(|\borders{\mathcal{T}}|)$.}

\begin{algorithm}
 \caption{\bordercandidates}\label{alg:bc}
\begin{algorithmic}[1]
\small
\STATE{\textbf{Input}:  A $Q$-tree $\mathcal{T}$}
\IF{$\depth(\mathcal{T}) = 0$}
\RETURN{$\mathcal{T}$}
\quad \COMMENT{Return the single element from the tree}
\ELSE
\RETURN{$\bordercandidates(\mathcal{T}^L)\cup \bordercandidates(\mathcal{T}^R)$}
\ENDIF
\STATE{\textbf{Output}: The set $\borders{\mathcal{T}}$ of border candidates of $\mathcal{T}$}
\end{algorithmic}
\end{algorithm}%

\subsubsection{Computing the minimum pairwise dissimilarity among trees}

\label{sssec:Tree_Dissimilarity} 
 

\sac{
A second key step is to define an appropriate dissimilarity between trees. This allows us to solve the problem recursively by decreasing the number of objects we need to sort in each iteration. In the initial case we have that this dissimilarity corresponds to $\mathbf{d}$, and since it is strict circular Robinson,} \partition (available in \Cref{app:subroutines}) with input $(\mathbf{T}, \mathbf{d})$ returns an arc partition of $\mathbf{T}$ stored in tuples. The elements in each tuple $\alpha$ must be consecutive. Therefore, for each $\alpha$ we build a new tree $\mathcal{T}_\alpha$ with a $Q$-node in the root whose $i$-th \sac{child} corresponds to $\alpha(i)$. \cgp{The recursion works by repeating} the process over the smaller family of trees $\mathbf{T}^\prime = \{\mathcal{T}_\alpha\}$, \cgp{until}  $|\mathbf{T}^\prime|=1$. 


\sac{For this to work we need to define a dissimilarity $\mathbf{d}^\prime$ over $\mathbf{T}^\prime$ in a way that sorting $\mathbf{T}^\prime$ with respect to $\mathbf{d}^\prime$, help us in our goal of sorting $\mathbf{T}$ with respect to $\mathbf{d}$. We now present a dissimilarity that does exactly that:} given two trees $\mathcal{T}_1,\mathcal{T}_2$, consider the dissimilarity $\dmin(\mathcal{T}_1,\mathcal{T}_2)\defined\min \{ \mathbf{d}(x,y) : x\in\leaves{\mathcal{T}_1},y\in\leaves{\mathcal{T}_2}\}$. Also let $ \argdmin (\mathcal{T}_1,\mathcal{T}_2)$ be the collection of all minimizers of this problem in $\leaves{\mathcal{T}_1}\times  \leaves{\mathcal{T}_2}$. The following \sac{Lemma justifies this choice, as it ensures that} by 
sorting $\mathbf{T}^\prime$ we obtain \cgp{a quasi-order among the elements of} 
our original set $\mathbf{T}$.

\begin{defi}[Quasi-order\footnote{This extends the definition introduced in \cite{chepoi1997recognition} for linear orders to cyclic orders.}]\label{defi:quasiorder}
Let $(\X,\mathscr{C})$ be a cyclically ordered set. An ordered partition $\{ A_{0}, \dots, A_{m-1} \}$ is a (consistent) cyclic quasi order if for all $(i,j,k)\in \mathscr{C}_m$, $x\in A_i$, $y\in A_j$ and $z\in A_k$ we have that $(x,y,z)\in \mathscr{C}$.

\end{defi}
\begin{lemma}\label{lema:dmintree}
Let $D\in \mathbb{R}^{n\times n} $ be a (strict) circular Robinson dissimilarity and let $\{A_i\}_{i\in [m]}$ be a cyclic quasi-order in $[n]$. The matrix $\Dmin(A_i,A_j)\defined \min \{D(k,l) :k\in A_i, l\in A_j \}$ is a (strict) circular Robinson dissimilarity.
\end{lemma}

Since nearest neighbours must be consecutive, the family of $Q$-trees $\mathbf{T}$ at each recursion is guaranteed to satisfy that $\{\leaves{\mathcal{T}}\}_{\mathcal{T}\in \mathbf{T}}$ is an arc partition. Therefore, by the previous Lemma we have that $D^\prime (i,j)\defined  \dmin (\mathcal{T}_i, \mathcal{T}_j)\in \presC$ ($\preC$) whenever the original dissimilarity matrix $D$ is in $\presC$ ($\preC$) and a Robinson ordering for $D^\prime$ yields a quasi order for $D$. 


 







\sac{A naïve computation of $\dmin$ does not lead to a global $\mathcal{O}(n^2)$ time complexity. This is why our next goal is to find a workaround.} The next result, which is a direct consequence of \Cref{unimod} in \Cref{apx:proofs:circ}, implies that we can reduce this search by only considering border candidates of each tree. \cgp{In particular, $\argdmin (\mathcal{T}_1,\mathcal{T}_2)\subseteq \borders{\mathcal{T}_1}\times \borders{\mathcal{T}_2}.$}


\begin{lemma}\label{lemma:borders}
  Suppose $D\in \scircular$ and let $\mathcal{I} = [a,b]$ be an arc in $([n], \mathscr{C}_n)$. Then \sac{for every $i\in \mathcal{I}^c$}, all minimizers of $\min \{D(i,j):j\in \mathcal{I}\}$ are contained in $ \{a, b\}.$
    %
\end{lemma}



\begin{obs}\label{obs:complexitydmin}

\sac{Let \dissimilarity be the procedure which receives a pair $(\mathcal{T}_1,\mathcal{T}_2)$ of $Q$-trees and a dissimilarity $\mathbf{d}$ to return $\dmin(\mathcal{T}_1,\mathcal{T}_2)$ and $ \argdmin (\mathcal{T}_1,\mathcal{T}_2)$ by brute force comparisons among border candidates.  It takes $\mathcal{O}(|\borders{\mathcal{T}_1}|\cdot |\borders{\mathcal{T}_2}|)$ operations and $\mathcal{O}(1)$ space, since by \cref{lemma:borders}, $|\displaystyle \argdmin (\mathcal{T}_1,\mathcal{T}_2)|\leq 4$. We omit its pseudocode for brevity.}
\end{obs}

Now we introduce the main procedure required for orienting $Q$-nodes within the trees.
\subsubsection{Consecutive $Q$-nodes orientation}
\label{sssec:Consecutive_Orientation}

Recall from \cref{coro:brutecorrect} that in order to orient the root of a $Q$-tree $\mathcal{T}_2$, it suffices to find border candidates for the arc $\leaves{\mathcal{T}_2}$. \cgp{Let ${\cal T}_1,{\cal T}_2,{\cal T}_3$ be children of a $Q$-node $\alpha$, such that 
$\mathcal{T}_2$ succeeds  $\mathcal{T}_1$ and precedes  $\mathcal{T}_3$, then we decide their orientation with the
\orient procedure (\cref{alg:orient}).}

\begin{algorithm}
 \caption{\orient}\label{alg:orient}
\begin{algorithmic}[1]
\small
\STATE{\textbf{Input}: Three consecutive subtrees $(\mathcal{T}_1,\mathcal{T}_2,\mathcal{T}_3)$ of a $Q$-node $\alpha$}
\STATE{Let $\mathcal{A}^\prime \defined   \borders{\mathcal{T}_1}$, $\mathcal{B}^\prime \defined  \borders{\mathcal{T}_3}$, $\mathcal{A}\defined  \lborders{\mathcal{T}_2}$ and $\mathcal{B} \defined  \rborders{\mathcal{T}_2}$. }
\STATE{\cgp{$x=\brute(\mathcal{A}^\prime,\mathcal{A},\mathcal{B},\mathcal{B}^\prime)$}}
\IF{ $x=$ \texttt{correct}}
    \STATE{Fix \sac{ the root} of $\mathcal{T}_2$ \sac{in $\alpha$}}
\ELSIF{ $x=$ \texttt{reverse}}
 \STATE{Reverse \sac{the root} of $\mathcal{T}_2$, then fix it \sac{in $\alpha$}}
\ELSE
\STATE{Label $\mathcal{T}_2$'s root as non-orientable and continue the algorithm as if the root of $\mathcal{T}_2$ \sac{were} fixed. This node is an actual $Q$-node of the tree of Robinson orderings.}
\ENDIF
\STATE{\textbf{Result:} \sac{The direct children of the root of $\mathcal{T}_2$ had been directly connected to $\alpha$}}
\end{algorithmic}
\end{algorithm}%

Since $\mathcal{A}^\prime \defined   \borders{\mathcal{T}_1}$, $\mathcal{B}^\prime \defined  \borders{\mathcal{T}_3}$, $\mathcal{A}\defined  \lborders{\mathcal{T}_2}$ and $\mathcal{B} \defined  \rborders{\mathcal{T}_2}$ are border candidates for the arc $\leaves{\mathcal{T}_2}$, the correctness of the procedure is due to the correctness of \cref{alg:bruteforce}. By \cref{obs:complexitybrute}, the complexity is given by $\mathcal{O}(n\cdot\max\{|\borders{\mathcal{T}_1}|,|\borders{\mathcal{T}_2}|,|\borders{\mathcal{T}_3}|\} )$.

\sac{In a computational implementation of the algorithm, $Q$-nodes can be represented by tuples, whereas $Q$-trees are simply nested tuples. For clarity, we present an example of the $Q$-node \textit{fixing} procedure mentioned in \cref{alg:orient}.}

\sac{In this example we consider $Q$-node $\alpha = ((0, 1, 2) ,(6, 5, (3,4)) , (7, (8,9,10)))$ taking values over $\X=[11]$. There are three $Q$-trees, which are consecutive in $\alpha$, which are $\mathcal{T}_1 = (0, 1, 2)$, $\mathcal{T}_2 = (6, 5, (3,4))$ and $\mathcal{T}_3 = (7, (8,9,10))$. An example of a possible outcome of $\orient(\mathcal{T}_1,\mathcal{T}_2,\mathcal{T}_3)$ would be that the root of $\mathcal{T}_2$ is reversed and fixed into $\alpha$. In such case $\alpha$ is modified into $((0, 1, 2) ,(3,4), 5, 6 , (7, (8,9,10)))$.}

\subsubsection{Complete internal orientation of a connected component and final orientation}
\label{sssec:Complete_Internal_Orientation}
A complete internal orientation of a \sac{$Q$-tree $\mathcal{T}$, is a process in which we determine the orientation of all nodes present in $\mathcal{T}$ except from those present in $\mathcal{T}^L$ and $\mathcal{T}^R$}. For this task we propose the procedure \complete (\cref{alg:complete}).

\begin{algorithm}
\small
 \caption{\complete}\label{alg:complete}
\begin{algorithmic}[1]
\STATE{\textbf{Input}: A $Q$-tree $\mathcal{T}$ with root $\alpha$ \cgp{and} $\depth(\mathcal{T})>1$}
\STATE{For $i\in[k]$, let $\alpha(i)$ be the $i$-th children of $\alpha$ and let $\mathcal{T}_{i}$ be the subtree whose root is $\alpha(i)$}
\STATE{\{\sac{Notice that ${\cal T}_{0}={\cal T}^L$ and ${\cal T}_{k-1}={\cal T}^R$}\}}
\WHILE{ $\exists\,\,1\leq j\leq  k-2$ such that the node $\alpha(j)$ has not been fixed into $\alpha$}
\FOR{$i=1\dots k-2$}
\STATE{$\orient({\cal T}_{i-1},{\cal T}_i,{\cal T}_{i-1})$}
\ENDFOR
\ENDWHILE
\STATE{\textbf{Result:} All elements in \cgp{$\leaves{\mathcal{T}}\setminus (\leaves{\mathcal{T}^L}\cup\leaves{\mathcal{T}^R})$} are directly connected to $\alpha$.}
\end{algorithmic}
\end{algorithm}%

The \sac{orientation over all the nodes} is done in a breadth-first search fashion\footnote{This way, trees of same depth are compared in \cref{alg:orient} (excluding comparisons with $\mathcal{T}^L$ and $\mathcal{T}^R$).}. Once again, the correctness of the procedure is due to the correctness of \cref{alg:bruteforce}. As an example, consider \cref{fig:examplecomplete}. Here, the tree at the top was constructed at the second recursion of the algorithm and represents a connected component of the nearest-neighbours graph over a family of $Q$-nodes. $\mathcal{T}^L$ corresponds to the tree with root in $Q_1$ and $\mathcal{T}^R$ corresponds to the tree with root in $Q_k$. The tree in the bottom corresponds to the tree after \cref{alg:complete}.

In the final recursion of \main we obtain a unique connected component from \partition, from which we construct a unique tree $\mathcal{T}$ such that $\leaves{\mathcal{T}}=\X$. Here, the cyclic order of $\X$ implies that the subtrees $\mathcal{T}^L$ and $\mathcal{T}^R$ are consecutive. Hence, to orient these trees we make a slight variation in the procedure \completefinal (\cref{alg:completefinal}). This is equivalent to consider the $Q$-node as a ring rather than as a list.

\begin{algorithm}
 \caption{ \completefinal}\label{alg:completefinal}
\begin{algorithmic}[1]
\small
\STATE{\textbf{Input}: A $Q$-tree $\mathcal{T}$ with root $\alpha$.}
\STATE{Let $\alpha(i)$ be the $i$-th children of $\alpha$ and let $\mathcal{T}_{i}$ be the subtree whose root is $\alpha(i)$}
\WHILE{$\exists j \in [|\alpha|]$ such that the node $\alpha(j)$ has not been fixed into $\alpha$}

\IF{$|\alpha|>2$}
\STATE{Let $\mathcal{T}_{-1}\defined  \mathcal{T}^R$ and $\mathcal{T}_{|\alpha|}\defined  \mathcal{T}^L$}
\FOR{$i\in [|\alpha|]$}
\STATE{$\orient(\mathcal{T}_{i-1}, \mathcal{T}_{i}, \mathcal{T}_{i+1})$}
\ENDFOR
\ELSIF{$|\alpha|=2$}
\STATE{$\orient(\mathcal{T}_1^R,\mathcal{T}_0,\mathcal{T}_1^L)$} 
\STATE{Fix $\alpha(1)$ into $\alpha$}
\ENDIF
\ENDWHILE
\STATE{\textbf{Result:} All $Q$-nodes in $\mathcal{T}$ are oriented.}
\end{algorithmic}
\end{algorithm}%

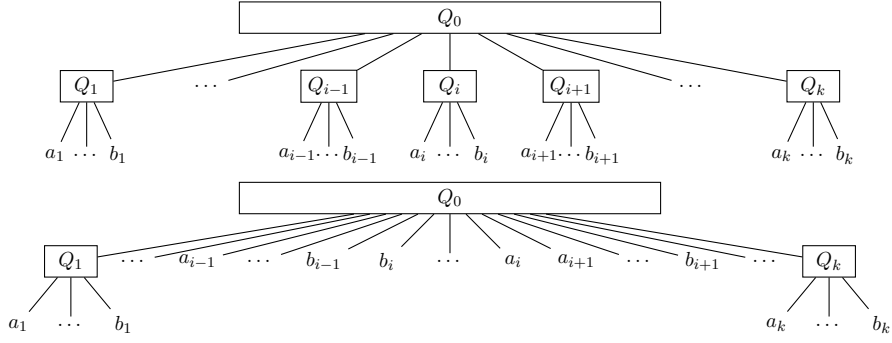
\begin{figure}
    \centering

        {\begin{tabular}[b]{c}%

         \begin{tikzpicture}[scale=0.7,transform shape,
  level distance=1.3cm,
  level 1/.style={sibling distance=2.3cm},
  level 2/.style={sibling distance=0.6cm},
  level 3/.style={sibling distance=0.6cm}]
  \node[rectangle,minimum width=8cm, draw] {$Q_0$}
    child{node [rectangle,minimum width=1cm, draw] {$Q_1$}
        child {node {$a_1$}}
         child{node{$\dots$}}
        child {node {$b_1$}}
        }
    child{node{$\dots$}}
    child{node [rectangle,minimum width=1cm, draw] {$Q_{i-1}$}
         child {node {$a_{i-1}$}}
          child{node{$\dots$}}
         child {node {$b_{i-1}$}}
        }
     child{node [rectangle,minimum width=1cm, draw] {$Q_{i}$}
         child {node {$a_{i}$}}
          child{node{$\dots$}}
         child {node {$b_{i}$}}
        }
     child{node [rectangle,minimum width=1cm, draw] {$Q_{i+1}$}
         child {node {$a_{i+1}$}}
          child{node{$\dots$}}
         child {node {$b_{i+1}$}}
        }
    child{node{$\dots$}}
    child{node [rectangle,minimum width=1cm, draw] {$Q_{k}$}
        child {node {$a_k$}}
         child{node{$\dots$}}
        child {node {$b_k$}}
        };
\end{tikzpicture}\\

         \begin{tikzpicture}[scale=0.7,transform shape,
  level distance=1.2cm,
  level 1/.style={sibling distance=1.2cm},
  level 2/.style={sibling distance=1cm},
  level 3/.style={sibling distance=1cm}]
  \node[rectangle,minimum width=8cm, draw] {$Q_0$}
    child{node [rectangle,minimum width=1cm, draw] {$Q_1$}
        child {node {$a_1$}}
         child{node{$\dots$}}
        child {node {$b_1$}}
        }
    child{node{$\dots$}}
    child {node {$a_{i-1}$}}
    child{node{$\dots$}}
    child {node {$b_{i-1}$}}
    child {node {$b_{i}$}}
    child{node{$\dots$}}
    child {node {$a_{i}$}}
    child {node {$a_{i+1}$}}
    child{node{$\dots$}}
    child {node {$b_{i+1}$}}
    child{node{$\dots$}}
    child{node [rectangle,minimum width=1cm, draw] {$Q_{k}$}
        child {node {$a_k$}}
         child{node{$\dots$}}
        child {node {$b_k$}}
        };
\end{tikzpicture}
\end{tabular}}
    \caption{Example of a connected component of the nearest neighbours graph at the second recursion of \main before and after \complete.}\label{fig:examplecomplete}

\end{figure}

\subsubsection{External orientation of trees}
\label{sssec:External_Orientation}
Since for each $\mathcal{T} \in \mathbf{T}$, the set $\leaves{\mathcal{T}}$ corresponds to an arc and as a consequence of \cref{lemma:borders}, $\dmin (\mathcal{T},\mathcal{T}^\prime)$ is attained at some $x\in \borders{\mathcal{T}}$ and $y\in \borders{\mathcal{T}^\prime}$ which are guaranteed to be borders of $\leaves{\mathcal{T}}$ and $\leaves{\mathcal{T}^\prime}$, respectively. Therefore, we must arrange some of their internal nodes in a way that $x$ and $y$ lie at the borders. \cgp{With this purpose in mind,} we propose the procedure \external (\cref{alg:external}).

\begin{algorithm}
 \caption{\external}\label{alg:external}
\begin{algorithmic}[1]
\small
\STATE{\textbf{Input}: $Q$-trees $\mathcal{T}$ and $\mathcal{T}^\prime$. The set $\argdmin (\mathcal{T},\mathcal{T}^\prime)$}
\FOR{$(x,y)\in \argdmin (\mathcal{T},\mathcal{T}^\prime)$}
\IF{$x\in \lborders{\mathcal{T}}$}
\STATE{}
\COMMENT{Fix every $Q$-node in $\mathcal{T}^L$ containing $x$ as a descendant from the root until the $Q$-node $\alpha$ where $x$ lies in a way such that $x$ is placed on the left}
\STATE{$\mathcal{J} = \mathcal{T}^L$}
\WHILE{$\depth(\mathcal{J})>0$}
\IF{$x\in \lborders{\mathcal{J}}$}
\STATE{Fix the root of $\mathcal{J}$ into its parent}
\ELSE
\STATE{Reverse and then fix the root of $\mathcal{J}$ into its parent}
\ENDIF
\STATE{$\mathcal{J} \leftarrow \mathcal{J}^L$}
\ENDWHILE 
\ELSE
\STATE{\cgp{Proceed analogously, with ${\cal B}^R({\cal T})$, instead of ${\cal B}^L({\cal T})$; and placing $x$ at the right, instead of the left}}
\ENDIF
\ENDFOR
\STATE{Repeat the same procedure with $\mathcal{T}^\prime$ and $y$}
\STATE{\textbf{Result:} Either all $Q$-nodes in $\mathcal{T}^L$ or $\mathcal{T}^R$ (resp. $\mathcal{T}^{\prime L}$ or $\mathcal{T}^{\prime R}$) are oriented}
\end{algorithmic}
\end{algorithm}%

An important observation is that in the tree $\mathcal{T}$ resulting from the first part of this procedure we have that $\mathcal{T}^L = \{x\}$ (assuming for simplicity that $x\in \lborders{\mathcal{T}}$). In the second part, we execute \complete (\cref{alg:complete}) with input $\mathcal{T}$. Since $\mathcal{T}^L = \{x\}$ at the end the only $Q$-nodes remaining to be oriented are the ones present in $\mathcal{T}^R$. As an example, we consider the $Q$-trees in \cref{fig:externalorientation}. Let $\mathcal{T}$ be the tree in \cref{fig:nspoon}. In this example, $\mathcal{T}^L$ is the subtree with root in $Q_1$ and $\mathcal{T}^R$ is the singleton $\{b_0\}$. Suppose by computing $\dmin (\mathcal{T},\mathcal{T}^\prime)$ for some other $\mathcal{T}^\prime$ we get that $\dmin$ is attained at $b_2\in \lborders{\mathcal{T}}$. In that case, we must fix $\mathcal{T}^L$ following the algorithm. Since $b_2$ is a left border in $Q_1$, this node is correctly oriented. However, since $b_2$ appears in the right of $Q_2$, we must reverse $Q_2$ as in \cref{fig:nspoonoriented1}. The resulting tree is the one in \cref{fig:nspoonoriented2}. Next, we perform a complete orientation and the resulting tree is the one in \cref{fig:nspoonoriented3}.

Notice that excluding the running time of \cref{alg:complete}, the number of operations required for this procedure is bounded by  $\mathcal{O}(\depth(\mathcal{T}^L))$. 

\begin{obs}\label{obs:logborder}
If $\mathcal{T}$ is a tree built at the $k$-th recursion of the algorithm, then clearly $\depth(\mathcal{T})\leq k$. We claim that since after this process either $\mathcal{T}^L$ or $\mathcal{T}^R$ gets completely oriented, then it holds that $|\borders{\mathcal{T}}| \leq k+1$.
We prove this by induction on $k$.
Notice that if $\mathcal{T}$ is composed by a single $Q$-node in the root, then $|\borders{\mathcal{T}}|\leq 2$. Now let $\mathcal{T}$ be a tree instantiated at the $k$-th recursion of the algorithm. W.l.o.g. assume $\mathcal{T}^L$ gets completely oriented. Then $\borders{\mathcal{T}} = \borders{\mathcal{T}^R}\cup \{x\}$. Hence, $|\borders{\mathcal{T}} | = |\borders{\mathcal{T}^R}|+1$. The claim follows by 
inducting on $\mathcal{T}^R$.
\end{obs}


\begin{figure}
    \centering

    \subfloat[Step 1]{\label{fig:nspoon}
         \begin{tikzpicture}[scale=0.4,transform shape,
  level distance=1.7cm,
  level 1/.style={sibling distance=3cm},
  level 2/.style={sibling distance=2cm},
  level 3/.style={sibling distance=1cm}]
  \node[rectangle,minimum width=5cm, draw] (r1) {\huge $Q$}
        child{node [rectangle,minimum width=2.5cm, draw] {\huge $Q_{1}$}
            child{node [rectangle,minimum width=1.4cm, draw] {\huge $Q_{2}$}
                child{node [rectangle,minimum width=1cm, draw] {\huge $Q_{3}$}
                    child {node{\huge $a_3$}}
                    child {node{$\dots$}}
                    child {node {\huge $b_3$}}}
                child {node{$\dots$}}
                child {node{$\dots$}}
                child {node{\huge $b_2$}}}
            child {node{$\dots$}}
            child {node{$\dots$}}
            child {node {\huge $b_1$}}}
    child{node{$\dots$}}
    child {node{$\dots$}}
    child {node{\huge $b_0$}};
\end{tikzpicture}}
    \qquad
    \subfloat[Step 2]{\label{fig:nspoonoriented1}

         \begin{tikzpicture}[scale=0.4,transform shape,
  level distance=1.7cm,
  level 1/.style={sibling distance=3cm},
  level 2/.style={sibling distance=2cm},
  level 3/.style={sibling distance=1.5cm}]
  \node[rectangle,minimum width=5cm, draw] (r1){\huge$Q$}
            child{node [rectangle,minimum width=1.4cm, draw] {\huge$Q_{2}$}
                child {node{\huge $b_2$}}
                child {node{$\dots$}}
           child{node [rectangle,minimum width=1cm, draw] {\huge$Q_{3}$}
                child {node{\huge $a_3$}}
                child {node{$\dots$}}
                child {node{\huge $b_3$}}}
                }
            child {node{$\dots$}}
    child {node{\huge $b_1$}}
    child {node{$\dots$}}
    child {node{\huge $b_0$}};
\end{tikzpicture}

    }
      \qquad
    \subfloat[Step 3]{ \label{fig:nspoonoriented2}

      \begin{tikzpicture}[scale=0.4,transform shape,
  level distance=1.7cm,
  level 1/.style={sibling distance=2cm},
  level 2/.style={sibling distance=1cm},
  level 3/.style={sibling distance=1cm}]
  \node[rectangle,minimum width=5cm, draw] (r1){\huge $Q$}
        child {node{\huge $b_2$}}
        child {node{$\dots$}}
            child{node [rectangle,minimum width=1cm, draw] {\huge $Q_{3}$}
                child {node{\huge $a_3$}}
                child {node{$\dots$}}
                child {node{\huge $b_3$}}}
        child {node{$\dots$}}
        child {node{\huge $b_1$}}
        child {node{$\dots$}}
        child {node{\huge $b_0$}};
\end{tikzpicture}
    }
      \qquad
    \subfloat[Step 4]{    \label{fig:nspoonoriented3}

         \begin{tikzpicture}[scale=0.4,transform shape,
  level distance=1.7cm,
  level 1/.style={sibling distance=1.5cm},
  level 2/.style={sibling distance=1cm},
  level 3/.style={sibling distance=1cm}]
  \node[rectangle,minimum width=5cm, draw] (r1){\huge $Q$}
        child {node{\huge $b_2$}}
        child {node{$\dots$}}
        child {node{\huge $a_3$}}
        child {node{$\dots$}}
        child {node{\huge $b_3$}}
        child {node{$\dots$}}
        child {node{\huge $b_1$}}
        child {node{$\dots$}}
        child {node{\huge $b_0$}};
\end{tikzpicture}
    }

    \caption{\cgp{ \external over tree $\mathcal{T}$ in which $\dmin (\mathcal{T},\mathcal{T}^\prime)$ is attained at $b_2\in \lborders{\mathcal{T}}$.}}\label{fig:externalorientation}

\end{figure}

\subsection{Analysis of the recursive seriation algorithm}

\label{subsec:Analysis_Rec_Seriation}

\begin{theorem}
Given $D\in \presC$, let $\mathcal{T}$ be the $PQ$-tree obtained from \cref{alg:Rec_seriation} with input $D$. Let $S(\mathcal{T})$ the \sac{set of all orderings} of $\X$ (permutations) represented by the tree. Then, $S_\scircular(D) =\dihedral_n \circ S(\mathcal{T})$, i.e. it solves the strict circular seriation problem.
\end{theorem}
\begin{proof}[Proof Sketch]

\sac{For simplicity, suppose we omit the orientation steps in \cref{alg:Rec_seriation} and leave them to the end of the process.} This does not affect the set of solutions but may increase the time complexity. Denote $\mathcal{T}^{pre}$ and $\mathcal{T}$ the trees before and after orientation, respectively. Also let $\mathbf{T}_k$ the family of trees instantiated at the $k$-th recursive step. Notice that by \cref{lema:dmintree}, evaluating $\dmin$ over $\mathbf{T}_k$ yields a dissimilarity matrix $D_k\in \presC$. 
Due to \cref{prop:dfs}, we have that $S_\scircular(D) \subset \dihedral_n \circ  S(\mathcal{T}^{pre})$ (at least all Robinson orderings are considered at this point). To complete the proof, it remains to show that in $\mathcal{T}$ all orientable $Q$-nodes originally in $\mathcal{T}^{pre}$ had been correctly fixed. To see this notice that the orientation of each $Q$-node in $\mathcal{T}^{pre}$ is tested either by \external (\cref{alg:external})
or \orient (\cref{alg:orient}). The correctness of \external is due to \cref{lemma:borders}. The correctness of \orient is due to \cref{coro:brutecorrect}.
\end{proof}

\begin{theorem}
\main runs in $\mathcal{O}(n^2)$ time.
\end{theorem}

 \begin{proof}

We count the number of operations required by the procedure \complete and \completefinal separately from the rest. At the $i$-th recursion let $\mathbf{T}(i)$ be the input $Q$-trees, let $k(i)\defined |\mathbf{T}(i)|$ and let $b(i) = \max_{\mathcal{T}\in \mathbf{T}(i)} |\borders{\mathcal{T}}|$. Then, by \cref{obs:complexitydmin}, computing $\dmin$, takes $\mathcal{O}(k(i)^2\cdot b(i)^2)$ operations. By \cref{obs:logborder}, the complexity of the procedure \external takes $\mathcal{O}(k(i)^2)$ operations. 
Computing $\NNG$ takes $\mathcal{O}(k(i)^2)$ operations. The procedure \depthfirst takes $\mathcal{O}(k(i))$ operations.

On the other hand, notice that in each step of the recursion, every tree is merged to its nearest neighbour. This implies that $k(i)\leq \frac{n}{2^i}$ and, therefore, the depth of the recursion is bounded by $\log_2(n)$. Since by \cref{obs:logborder} $b(i)\leq i+1$ then, there is some constant $C_1>0$ such that the the total number of operations of this procedure is bounded by $C_1 \sum_{i=0}^{\log_2(n)} \left(\frac{n}{2^i}\right)^2 \left(i+1\right)^2 + \left(\frac{n}{2^i}\right)^2+ \left(\frac{n}{2^i}\right) = {\cal O}(n^2)$.

It remains to consider \complete and \completefinal. In this procedures, all $Q$-nodes $\alpha$ are oriented through \orient (\cref{alg:orient}) with input $(\mathcal{T}_1,\mathcal{T}_2,\mathcal{T}_3)$ where $\alpha$ is the root of $\mathcal{T}_2$. To count the operations of this procedure we consider two cases. The first (and most common) case is when $\mathcal{T}_1,\mathcal{T}_2$ and $\mathcal{T}_3$ are trees instantiated at the same recursive step. In this case, if they \sac{were} instantiated at the $i$-th recursion then by \cref{obs:logborder} and \cref{obs:complexitybrute} the orientation takes $\mathcal{O}((i+1)\cdot n)$ operations.

 By counting on the recursion where each node was instantiated, the total number of operations involving first case $Q$-nodes can be bounded by $C_2\cdot \sum_{i=0}^{\log_2(n)} \left(\frac{n}{2^i}\right)(i+1)\cdot n = {\cal O}(n^2). $

 A second case to consider is during the complete internal orientation of a connected component. Let $\mathcal{T}$ be a tree generated from a connected component of the nearest-neighbours graph at the $i$-th recursion of the algorithm. Then, in \complete (or \completefinal) with input $\mathcal{T}$, some of the internal $Q$-nodes will be oriented by having as border candidates $\borders{\mathcal{T}^L}$ and $\borders{\mathcal{T}^R}$. Since $\depth(\mathcal{T})\leq i$, this can occur for $\borders{\mathcal{T}^L}$ (resp. $\borders{\mathcal{T}^R}$) for at most $i$ internal $Q$-nodes of $\mathcal{T}$.
\cgp{Let} $C(i)$ be the number of connected component found in the $i$-th recursion, then the number of second case $Q$-nodes is at most $C(i)\cdot i \cdot 2$. Since by \cref{obs:logborder}, $|\borders{\mathcal{T}^L}|\leq i$ and $|\borders{\mathcal{T}^R}|\leq i$, the number of operations required for orienting all this nodes is bounded by $\mathcal{O}(C(i)\cdot i^2\cdot n)$. Again, by counting through the recursion levels and considering that $C(i)\leq n/2^{i}$, the total cost of orienting second case $Q$-nodes is bounded by $C_3\cdot \sum_{i=0}^{\log_2(n)}
\left(\frac{n}{2^i}\right) \cdot i^2 \cdot n = {\cal O}(n^2),
$
which proves the result.
 \end{proof}

\subsection{$PQ$-tree of solutions in the strict Robinson case}

It is clear that if a sequence is strictly monotone the only permutation that preserves this property is the one that reverses the sequence. Therefore if $D\in \slinear$, we have that $S_{\slinear}(D)=\{\textbf{e}, \textbf{r} \}\cong \dihedral_1$. However, it is not immediately clear which permutations are the ones that preserve strict unimodality. The next Lemma will let us conclude that there is at most one \textit{non trivial} ordering for $D\in \scircular$.

\begin{lemma}\label{lemma:nontrivial}
Let $D\in \scircular$ and let $\mathcal{I}_1,\dots \mathcal{I}_k$ be disjoint arcs of $[n]$. Let $\sigma_{\mathcal{I}_i}$ be the permutation that reverses $\mathcal{I}_i$. Then, at most one of the $\sigma_{\mathcal{I}_i}$'s produces a new Robinson ordering.


\end{lemma}
\begin{proof}
Suppose $\sigma_\mathcal{I}$ is a Robinson ordering for some arc $\mathcal{I}$. For every $i\in [n]$, let $M(i) =  \argmax_{j}D(i,j)$.  We claim that for every $i\notin \mathcal{I}$ it holds that $M(i)\subset \mathcal{I}$. Otherwise, given $m^*\in M(i)$, by the connectivity of $\mathcal{I}$, we must have that $\mathcal{I}$ must be strictly contained in one of the two paths connecting $i$ and $m^*$. Also notice that $D(i,\cdot)$ is strictly monotone in such path. Hence, reversing $\mathcal{I}$ would violate the monotonicity of such sequence (an thus the unimodality of the whole sequence). This proves the claim. 
Since $\mathcal{I}^c$ is an arc, by the same argument we have that $i\in \mathcal{I}$ implies $M(i)\subset \mathcal{I}^c$. Hence, the only arcs that can be reversed are $\mathcal{I}$ and $\mathcal{I}^c$.
\end{proof}


\section{Behavior for large \(n\)}\label{sec:generative}

The literature on the seriation problem has mostly focused on finite ordered sets, either linearly or cyclically ordered, on suitable classes of matrices encoding properties of this order, 
such as Robinson matrices, and on efficient algorithms for its solution. However, typically the use of seriation algorithms is motivated by the interpretation of data as embedded in a closed curve, and it is unclear how these combinatorial solutions relate to the underlying order of a continuous object.

To bridge this gap, we provide a simple generative model of sampling from a continuous and periodic structure. That sample, and more specifically the dissimilarities between pairs of points from the sample, will be the input of our strict seriation algorithm. The question we want to answer is: {\em to which extent the solution obtained by the seriation algorithm applied to a random sample reflects the underlying ordering of the periodic structure?} We will answer this question by proving that as the sample size $n$ grows, the expected Kendall-tau distance from the strict circular seriation algorithm solution to the order inherited from the continuous model decreases at a rate $\rate{n}$.

\subsection{Reduction to \(\S^1\)}

We will consider our periodic continuous structure as parameterized by the unit circle. Equivalently, we will use the set $[0,1)$ as the set of points, where we topologically identify $0$ and $1$, making it a circular-like structure. This set is endowed with the {\em natural cyclic order}, which results from embedding $\circInt$ into $\sphere$. 
We assume the set $\circInt$ is endowed with a dissimilarity \(\bfd\). We will make some assumptions that relate the circular ordering to the circular Robinson property.

\begin{axiom}\label{axiom:robinson}
     \(\bfd\) is continuous, and
     {\em strict circular Robinson}, i.e.,
    \begin{equation}
    \label{eq:dInfRobinson}
        \forall\,\, \mbox{cyclically ordered }x,y,z,w\in \circInt:\,\,
        \bfd(y,w) > \min\{\bfd(y,x), \bfd(y,z)\}. 
    \end{equation}
\end{axiom}

One natural question is how general this continuous model is. We claim that the assumption that our sample space is the unit circle is without loss of generality\footnote{From now w.l.o.g.}. For example, if the sample space is a one dimensional compact manifold of $\mathbb{R}^d$, we can parameterize the manifold by its arc-length $\gamma:[0,1)\mapsto\mathbb{R}^d$, and let $\bfd(t,s):=\|\gamma(t)-\gamma(s)\|$, which is clearly continuous. Notice however that the validity of the strict circular Robinson property is not guaranteed in this example: such assumption depends on the relative positions of points in space.

\subsection{Solutions in the limit}

To understand the set of solutions in the limit we first need to characterize the natural symmetries of the strict Robinson dissimilarity \(\bfd\). To do so, we consider the family of {\em cyclic shifts} \(\{\pi_s:\,s\in\circInt\}\) defined by \(\pi_s(t) = t + s\bmod 1\), and the {\em reversal} \(\pi_r(t) = 1 - t\). We let \(\dihedral_{\infty} := \gen{\pi_s, \pi_r:\, s\in\circInt}\). In addition, given an arc \(\mathcal{I} := (t,s)\subsetneq \circInt\), we let \(\sigma_{\mathcal{I}}\) be the bijection that reverses \(\mathcal{I}\) and fixes \(\mathcal{I}^c\). Since in the finite case all solutions can be expressed as compositions of such permutations, in the continuous case we look for solutions in $\syminf \defined \dihedral_\infty \circ \langle \sigma_\mathcal{I} : \mathcal{I} \text{ arc}\rangle $.



\begin{theorem}\label{th:continuo}
Suppose $\bfd$ satisfies \Cref{axiom:robinson}, and let
\(\pi\in \syminf\). If $\bfd\circ\pi$ is strict circular Robinson then $\pi\in \dihedral_{\infty}.$
\end{theorem}

This result can be seen as a well-posedness statement of the seriation problem in the continuous limit. Our next goal is to study its consequences 
for large (but finite) sample size.

\subsection{Approximate well-posedness of seriation in the large $n$ regime}
We now propose a sampling model from the continuous model. We uniformly at random extract a size $n$ sample from $\circInt$. We denote this sample by  \(\sample{n} := \{\x_0,\ldots, \x_{n-1}\}\). If we let $\lambda$ be the Lebesgue measure on $[0,1)$, then our sampling is distributed as $\lambda^n$. Let \(D_{\sample{n}}\) denote the dissimilarity matrix associated to \(\sample{n}\). In particular, if \(\x_0,\ldots, \x_{n-1}\) are cyclically ordered, then the dissimilarity matrix is strict circular Robinson (cf. \Cref{axiom:robinson}).

Despite that in the continuous case there is a unique Robinson ordering, with finitely many samples there might exist non-trivial orderings (cf. \cref{lemma:nontrivial}). In what follows we study conditions under which for a large sample, any ordering in $S_{\scircular}(\sampled{n})$ is close to the one induced by the curve. Our closeness measure
%
is given by the Kendall-tau's metric $\kendall$ 
and the goal is to bound the expected value of the diameter of the set of solutions:

\begin{defi}[Kendall-tau's metric \cite{{kendall1938new},ma2020optimal}]\label{defi:kendall}
We define the Kendall-tau distance between permutations $\pi_1$ and $\pi_2$ as
$\kendall \left(\pi_{1}, \pi_{2}\right)\defined
|\mathcal{G}\left(\pi_{1}, \pi_{2}\right)|/\binom{n}{2},$ 
where 
$\mathcal{G}\left(\pi_{1}, \pi_{2}\right)$ corresponds to the set of discordant pairs defined as 
{\small 
$$
\mathcal{G}\left(\pi_{1}, \pi_{2}\right)\defined\left\{(i, j): i<j,\left[\pi_{1}(i)<\pi_{1}(j) \wedge \pi_{2}(i)>\pi_{2}(j)\right] \vee\left[\pi_{1}(i)>\pi_{1}(j) \wedge \pi_{2}(i)<\pi_{2}(j)\right]\right\}.
$$}
\end{defi}
The denominator $\binom{n}{2}$ ensures that $\kendall\left(\pi_{1}, \pi_{2}\right) \in[0,1]$. 
The next definition of diameter takes into account that for seriation cyclic permutations provide the same ordering.

\begin{defi}
Given a set $S\subset \sym$, \sac{the diameter of $S$} is defined as $\diam(S)\defined  \max_{\pi_1,\pi_2\in S}\min_{\hat{\pi}_1\in \dihedral_n \circ \pi_1}\tau_K(\hat{\pi}_1, \pi_2)$.
\end{defi}

 Let $\operatorname{Arc}: \circInt\times \circInt\rightarrow[0,\frac{1}{2}]$ be the length of the shortest arc connecting two points in the unit circle, i.e.~$\operatorname{Arc}(\theta_1, \theta_2)=\min\{|\theta_1-\theta_2|,1-|\theta_1-\theta_2|\}$. To prove rates on the Kendall-tau distance we make a final assumption. This condition allows us to avoid making overly restrictive metric assumptions on the dissimilarity, but still enjoying a weaker form of distance.

\begin{axiom}\label{axiom:bilip}
The dissimilarity $\mathbf{d}$ satisfies the following bi-Lipschitz property:
\begin{equation} \label{eqn:biLipschitz}
(\exists L\geq \ell>0)(\forall s,t\in \circInt) \qquad \ell\cdot Arc(s,t)\leq \mathbf{d}(s, t)\leq L\cdot Arc(s,t).
\end{equation}
\end{axiom}



We conclude this Section by providing a rate on the expected Kendall-tau diameter of the set of solutions of the circular Robinson algorithm. Hence, all these solutions must be close to the underlying order of the continuous model. Its proof is deferred to \Cref{App:pf_gen_model}.
\begin{theorem}\label{th:main}
 Let $\sample{n} = \{\x_0, \x_1,\dots \x_{n-1} \} \iid \uniform$. Then given any $\bfd$ satisfying \cref{axiom:robinson} and \cref{axiom:bilip} we have that
  \begin{equation}
  \small
   {\textstyle
 \mathbb{E}_{\sample{n}}\left[ \diam (S_{\scircular}(D_{\sample{n}})) \right]  =  {\cal O}\Big( \frac{(L+\ell)}{\ell}\cdot\frac{\log(n)}{n} \Big)}.
 \end{equation}
 
\end{theorem}

\appendix

\section{Proofs}

\label{apx:proofs}

\subsection{Proof of~\Cref{prop:rob_is_unimod}}
\label{apx:proofs:circ}

We need two auxiliary results first.

\begin{prop}\label{unimod}
    \(f\) is unimodal (resp. strictly unimodal) if and only if for \(i \leq j \leq k\) we have \(f_j \geq \min\{f_i, f_k\}\) (resp. \(f_j > \min\{f_i, f_k\}\)).
\end{prop}

\begin{proof}
    Suppose \(f\) is unimodal, let \(m\) be a mode, and suppose there are \(i,j,k\), not all equal, such that \(i \leq j \leq k\) and \(x_j < \min\{x_i, x_k\}\). Then \(x_i > x_j\) and \(x_j < x_k\). This implies \(m \leq j\) and \(m \geq j\). Hence \(m = j\). This is a contradiction. Now, suppose \(f\) satisfies the inequality but has no mode. Then \(j\) is not a mode, and there is \(i < j\) and \(k > j\) such that \(f_i > f_j\) and \(f_k > f_j\). This is a contradiction. The proof for the strictly unimodal case follows from the same arguments.
\end{proof}

\begin{prop}\label{prop:cycleSeq}
    If \((i_{-1}, i_{0}, i_1), (i_0, i_1, i_2)\in \circO_n\) then for each \(k\in \{-1,1,2\}\) there is \(q_k \in [n]\) such that \(i_k = i_0 + q_k\bmod n\). Furthermore, \(q_{1} \leq q_2 \leq q_{-1}\).
\end{prop}

\begin{proof}
    Consider \(q_k = i_k - i_0\bmod n\). Then \(q_0 = 0\). Since cyclic shifts do not change cyclic orderings, this implies \(q_1 \leq q_2 \leq q_{-1}\). This proves the proposition.
\end{proof}

\begin{proof}[Proof of~\Cref{prop:rob_is_unimod}]
    For simplicity we define \(d^i_j \defined D(i, i + j\mod n)\). (2~\(\Rightarrow\)~1) From~\Cref{prop:cycleSeq} we can write \(i = j + q_i\), \(k = j + q_k\) and \(\ell = j + q_\ell\) with \(q_k \leq q_\ell \leq q_i\). Since \(d^j\) is unimodal, from~\Cref{unimod} we deduce \(d^j_{q_\ell} \geq \min\{d^{j}_{q_k}, d^j_{q_i}\}\). (1~\(\Rightarrow\)~2) If \(d^j\) is not unimodal, by~\Cref{unimod} there are \(q_k \leq  q_\ell \leq q_i\) with \(d^j_{q_\ell} < d^j_{q_k}\) and \(d^j_{q_\ell} < d^j_{q_i}\). If we define \(i= j + q_i\bmod n\), \(k = j + q_k \bmod n\) and \(\ell = j + q_\ell \bmod n\) we see that \((i, j, k), (j, k, \ell)\in\circO_n\). This contradicts 1.
\end{proof}

\subsection{Proofs for \Cref{sec:algorithm}}
\label{app:bruteiscorrect}

\begin{proof}[Proof of~\Cref{lema:bruteiscorrect}]

Let $z\in \X$ and denote $f_z(\cdot)\defined  \mathbf{d}(z,\cdot )$. First, notice that 
\begin{equation}\label{eq:balloverlap}
\begin{split}
       (\exists r>0). \{a,a^\prime\}\subset B_r(z) \,\,\wedge\,\, \{b,b^\prime\}\subset B_r(z)^c \\
   \Leftrightarrow \max \{ f_z(a), f_z(a^{\prime})\} <  \min \{ f_z(b), f_z(b^{\prime})\}.
   \end{split}
\end{equation}
(1.~$\Rightarrow$~2.) By \cref{eq:balloverlap}, this implication is direct from the fact that $a\in \mathcal{A}, b\in \mathcal{B},a^\prime\in \mathcal{A}^\prime$ and $b^\prime\in \mathcal{B}^\prime$. (2.$\Rightarrow$ 1.) Let $r\defined  \max \{ f_z(x), f_z(x^{\prime})\}$, thus $\{x,x^\prime\}\subset B_r(z)$ and  $\{y,y^\prime\}\subset B_r(z)^c$. By \cref{prop:arc_is_ball} this ball is an arc, and therefore is connected in any Robinson ordering. This implies that all elements in between $x$ and $x^\prime$ (in all Robinson orderings), including $a$ and $a^\prime$ must also be present in $B_r(z)$. Similarly, all elements in between $y$ and $y^\prime$, including $b$ and $b^\prime$ must not be present in $B_r(z)$. The implication follows from \eqref{eq:balloverlap}. (3.$\Rightarrow$ 2.) Direct. (2.$\Rightarrow$ 3.) Notice that given any $z\in \X$ and any $t>0$ we have that if there is some $a\in \mathcal{A}$ and $a^\prime \in \mathcal{A}^\prime$ such that $\max \{ f_z(a),f_z(a^\prime)\}<t$. Then, $f_z(a)<t \wedge f_z(a^\prime)<t$. Which implies $\max\{\min f_z(\mathcal{A}),\min f_z(\mathcal{A}^\prime)\}<t$.
Similarly, the existence of $b\in \mathcal{B}$ and $b^\prime \in \mathcal{B}^\prime$ such that $\min \{ f_z(b),f_z(b^\prime)\}>t$ implies that $\min\{\max f_z(\mathcal{B}),\max f_z(\mathcal{B}^\prime)\}>t$. This proves the final implication, and hence the result.
\end{proof}

\begin{proof}[Proof of~\Cref{lema:dmintree}]
\cgp{We only prove the strict case, as the non-strict case follows an analogous argument.} 
Let $x_d\in B_d$ and $x_b\in B_b$ be such that $\Dmin(B_b,B_d) = D(x_b, x_d)$, and let $x_c\in B_c, x_a\in B_a$ be arbitrary. We notice that $x_a,x_b, x_c,x_d $ is cyclically ordered, hence
\[
\Dmin(B_b,B_d) = D(x_b, x_d) > \min \{D(x_b, x_a), D(x_b, x_c)\}.
\]
On the other hand, 
\[
    \min \{D(x_b, x_a), D(x_b, x_c)\} \geq \min\{\Dmin(B_b,B_a), \Dmin(B_b,B_c)\}
\]
by definition of $\Dmin$, proving the result. 
\end{proof}

\subsection{Proof of \cref{th:continuo}}

\begin{proof}
Suppose by contradiction that there exist an arc $\mathcal{I}=[a,b)$ such that $\bfd\circ\sigma_{\mathcal{I}}$ is strict Robinson. For $\epsilon>0$ small, $a-\epsilon, a, b, b+\epsilon$ are cyclically ordered. By hypothesis,
\begin{equation}  \label{eqn:ineq_strict_continuous}
\bfd(\sigma(a) ,\sigma(b+\epsilon)) > \min \{\bfd(\sigma(a) ,\sigma(a-\epsilon)) , \bfd(\sigma(a), \sigma(b))\},
\end{equation}
and since $b+\epsilon , a-\epsilon \notin \mathcal{I}$, we get that $\sigma(a-\epsilon) = a-\epsilon$ and $\sigma(b+\epsilon) = b+\epsilon$. On the other hand, $\sigma(a)=b$ and $\sigma(b)=a$. Therefore, we can rewrite \eqref{eqn:ineq_strict_continuous} as
\begin{equation} \label{eqn:ineq_strict_continuous_2}
    \bfd(b ,b+\epsilon) > \min\{\bfd(b ,a-\epsilon) , \bfd(b ,a)\}.
\end{equation}
Let $\delta:=  \bfd(b ,a) >0$. By continuity we get that $\bfd(\sigma(a)) ,\sigma(b+\epsilon)) \rightarrow 0$ and $\bfd(b ,a-\epsilon)\rightarrow \delta$ as $\epsilon \rightarrow 0$. For sufficiently small $\epsilon$, this is a contradiction with \eqref{eqn:ineq_strict_continuous_2}.
\end{proof}

\subsection{Proof of \cref{th:main}} \label{App:pf_gen_model}
 Given $\x_0, \dots , \x_{n-1}\in \circInt$, the \textit{order statistics} correspond to the variables $\x_{(1)},\x_{(2)},\dots ,\x_{(k)}$ obtained by sorting the samples by increasing order. The \textit{gaps} of the sample correspond to the variables $w_i\defined \x_{(i+1)}-\x_{(i)}$. Let $\epsilon_n \defined \max_{i\in [n]} w_i$. The following result can be found in \cite[Theorem 1.2]{pinelis2019order}.

\begin{prop}\label{prop:probepsilon}
Suppose $\x_0, \x_1,\dots \x_{n-1} \iid \uniform$. Then, $$\textstyle\mathbb{P}(\epsilon_n \geq z )\leq \sum_{j=1}^{n+1}(-1)^{j-1} \binom{n+1}{j}(1-j z)_{+}^{n}$$
\end{prop}

 \begin{prop}\label{prop:epsilon}
Given any $\mathcal{I}=[\x_i,\x_j]$, we write $\mu(\mathcal{I})$ to denote $\operatorname{Arc}(\x_i,\x_j)$. Suppose that  \cref{axiom:robinson} and \cref{axiom:bilip} hold. Then the inequality $\epsilon_n<\constant$ implies that any arc $\mathcal{I}\subset \sample{n}$ such that $\mu(\mathcal{I})>\delta$ has a unique orientation in any circular Robinson ordering of $\sampled{n}$.
 \end{prop}
 
 \begin{proof} 
 Let $s\defined \x_{(i)}$ and $t\defined \x_{(j)}$ for some $i<j$. Let $\delta \in (0,\frac{1}{2})$ and consider the arc $\mathcal{I}=[s,t]$ in $\sample{n}$. Suppose $\epsilon_n<\constant$ and $\mu(\mathcal{I})>\delta$. Let $s^+=\x_{(i-1\bmod n)}$ and $t^+=\x_{(j+1\bmod n)}$. We claim that $B_{s}(\bfd (s,s^+))\between^*\mathcal{I}$. To prove the claim, it suffices to prove that 
\begin{equation}
   \mathbf{d}( s,s^+) <  \min \{ \mathbf{d}( s,t), \mathbf{d}( s,t^+)\}.
\end{equation}

First, notice that $\mathbf{d}( s,s^+)\leq L\cdot \epsilon_n$. Second, notice that since $\operatorname{Arc}(s,t)>\delta$, then
$$ \mathbf{d}(s,t^+)\geq \ell\cdot \operatorname{Arc}(s,t^+)
        \geq \ell\cdot  (\operatorname{Arc}(s,t) - \epsilon_n)
         \geq \ell\cdot  (\delta- \epsilon_n), $$
         
and therefore $ \min \{ \mathbf{d}( s,t), \mathbf{d}( s,t^+)\} \geq \ell\cdot  (\delta- \epsilon_n)$. Joining this two results with the fact that $\epsilon_n<\constant$ a proves the claim.
 \end{proof}

\begin{lemma}\label{lema:diameter}

 Let $n\geq \log(1/\delta)/\delta$, and let $\x_0, \x_1,\dots \x_{n-1} \iid \uniform$.  
 Let $E_n(\delta)=\{\epsilon_n < \constant\}$, then {\small $$
 \int_{E_n(\delta)} \diam S_{\scircular}(\sampled{n}(\omega) )\text{ d}\lambda^n(\omega) = \mathcal{O}\Big(\delta^2+\frac{\delta\log(1/\delta)}{n}\Big).$$}
\end{lemma}

\begin{proof}

Recall from \cref{lemma:nontrivial} that there is at most one non-trivial ordering $\sigma_\mathcal{I}$ of $\sampled{n}$ which corresponds to the permutation that reverses the arc $\mathcal{I}\cap \sample{n}$. Therefore, it suffices to bound the integral of the random variable $\kendall(\id,\sigma^*)$, where $\sigma^* \in \argmin_{\hat{\sigma}^\in \{\pi_\mathbf{r}\circ \sigma_\mathcal{I} , \sigma_\mathcal{I}\}}\kendall(\id,\hat{\sigma})$. 

The number of discordant pairs between $\sigma^*$ and $\id$ is bounded by $\frac12\min\{ |\mathcal{I}\cap \sample{n}|,n - |\mathcal{I}\cap \sample{n}|\}^2$. Hence, we will focus in bounding this expression.  
Let $\omega \in E_n(\delta)$. By \cref{prop:epsilon}, $\mu(\mathcal{I}(\omega))\leq \delta$. This implies that either $\lambda (\mathcal{I}) \leq \delta $ or $\lambda (\mathcal{I}) \geq 1 - \delta$. \cgp{Therefore,}
\begin{multline*}
    \textstyle\int_{E_n(\delta)}\min\{ |\mathcal{I}\cap \sample{n} |,n - |\mathcal{I}\cap \sample{n}|\}^2\text{ d}\lambda^n(\omega)
   \\ 
   \textstyle\leq  \int_{E_n(\delta)} \max_{{\cal J} \mbox{\tiny interval, }\lambda({\cal J})\leq \delta} |{\cal J}\cap\sample{n}|^2\text{ d}\lambda^n(\omega)
   \leq\int_{\Omega} \max_{\cal J\mbox{\tiny interval, }\lambda({\cal J})\leq \delta} |{\cal J}\cap\sample{n}|^2\text{ d}\lambda^n(\omega).
\end{multline*}
We bound the random variable inside the integral using a balls and bins argument. W.l.o.g. $1/\delta$ is an integer. Let $({\cal J}_{i})_{i\in[1/\delta]}$ be a partition of $\circInt$ by disjoint intervals of length $\delta$. For any $\omega\in \Omega$, the maximizer in the integral above lies in at most two of the partition intervals. Therefore,  $\max_{\lambda({\cal J})\leq \delta} |{\cal J}\cap\sample{n}|\leq 2\max_{i\in [n]}|{\cal J}_i\cap \sample{n}|$. Next, we can estimate $\max_{i\in [n]}|{\cal J}_i\cap \sample{n}|$ by looking into the problem of throwing $n$ balls into $1/\delta$ bins (see \cite{raab1998balls} for further details); since we further assumed that $n\geq 1/\delta\log(1/\delta)$, then w.h.p., the maximum occupancy is bounded by $n\delta+\Theta(\sqrt{n\delta\log(1/\delta)})$. Plugging this bound above yields the result.

\end{proof}

\begin{proof}[Proof of \cref{th:main}]

Let $E_n(\delta)$ denote the event $\{\epsilon_n < \constant\}$. Denote the random variable $Z= \diam (S_{\scircular}(\sampled{n}))$. Then, 
 \begin{equation}\label{eq:th1}
  \textstyle\mathbb{E}\left[  Z \right]   = \int_{E_n(\delta)}Z(\omega) \text{d}\lambda(\omega) +  \int_{E_n(\delta)^c} Z(\omega) \text{d}\lambda(\omega) \leq \mathcal{O}\Big(\delta^2+\frac{\delta\log(1/\delta)}{n}\Big)+  \mathbb{P}\left[  E_n(\delta)^c\right],
 \end{equation}
 where in the inequality we used \cref{lema:diameter} and $\kendall\leq 1$. By \cref{prop:probepsilon} we have  
 \begin{equation}\label{eq:probepsilon}
\textstyle\mathbb{P}\left[  E_n(\delta)^c\right] \leq \sum_{j=1}^{n+1}(-1)^{j-1} \binom{n+1}{j}(1-j x)_{+}^{n}  \leq \sum_{j=1}^{n+1} \left( \frac{e(n+1)}{j}\right)^j \exp\{-xjn\},
 \end{equation} 
where $x=\constant $. By taking $\delta(n)= (L+\ell)\cdot \log(e(n+1)^2)/(n\cdot \ell)$ we obtain that \cref{eq:probepsilon} can be bounded by $\sum_{j=1}^{n+1}\frac{1}{(n+1)^j}\in \mathcal{O}(1/n)$. By \cref{eq:th1} and \cref{eq:probepsilon} we conclude that $\mathbf{E}[Z]={\cal O}(\frac{L+\ell}{\ell}\frac{\log n}{n}).$
\end{proof}

\section{Auxiliary subroutines}

\label{app:subroutines}

\begin{algorithm}
 \caption{\partition }\label{alg:gac}
\begin{algorithmic}[1]
\small
\STATE{\textbf{Input}: A dissimilarity $\mathbf{d}$ and a set $\mathbf{T}$}
\STATE{$\mathcal{N}_G(x)\defined \{y\in \mathbf{T}:\, x\in \NN(y)\vee y\in \NN(x) \}$}
\COMMENT{Compute the neighbourhood function}
\STATE{$\mathcal{B}\defined \{x\in \mathbf{T} : |\mathcal{N}_G(x)|=1\}$\\}
\COMMENT{Find all degree $1$ nodes (if there are no such nodes pick any)}
\STATE{$i= 0$}
\COMMENT{Run DFS starting at every non visited degree $1$ node}
\FOR{$x\in \mathcal{B}\setminus \cup_{j< i} \alpha_i$}
\STATE{$\alpha_i = \DFS(\mathcal{N}_G, \emptyset, x)$}
\STATE{$i=i+1$}
\ENDFOR
\STATE{\textbf{Output}: An arc partition stored into tuples $\mathcal{P}\defined\{\alpha_i\}_{i\in [k]}$}
\end{algorithmic}
\end{algorithm}%

\begin{algorithm}
 \caption{\depthfirst}\label{alg:dfs}
\begin{algorithmic}[1]
\small
\STATE{\textbf{Input}:  The neighbourhood function of a graph $\mathcal{N}_G(\cdot)$, a tuple $\alpha$ of visited nodes and a starting node $x$ }
\STATE{$\alpha(n) = x$}
\COMMENT{Set $x$ as $n$-th visited node where $n$ is the size of $\alpha$}
\FOR{$y \in \mathcal{N}_G(x)\setminus \alpha$}
\STATE{$\alpha = \DFS( \mathcal{N}_G, \alpha, y)$}
\COMMENT{Recurse over all adjacent nodes that have not been visited}
\ENDFOR
\RETURN{$\alpha$}
\STATE{\textbf{Output}: A tuple of visited nodes $\alpha$}
\end{algorithmic}
\end{algorithm}%

\section*{Acknowledgments}
We would like to thank Alexandre d'Aspremont for valuable discussions at different stages of this work. 

\bibliographystyle{siamplain}
\bibliography{references}

\end{document}